\documentclass[11pt,a4paper]{article}
\usepackage{amsmath, amssymb, amsthm}
\usepackage{geometry}
\usepackage{graphicx}%
\usepackage{multirow}%
\usepackage{mathrsfs}%
\usepackage[title]{appendix}%
\usepackage{xcolor}%
\usepackage{textcomp}%
\usepackage{manyfoot}%
\usepackage{booktabs}%
\usepackage{algorithm}%
\usepackage{algorithmicx}%
\usepackage{algpseudocode}%
\usepackage{listings}%
\usepackage{physics}
\usepackage[numbers]{natbib}
\usepackage{bm}
\usepackage{esint}
\numberwithin{equation}{section}
\usepackage[colorlinks=true, , allcolors=blue]{hyperref}
\usepackage{soul}
\newtheorem{theorem}{Theorem}%
\newtheorem{proposition}[theorem]{Proposition}%
\newtheorem{cor}[theorem]{Corollary}

\newtheorem{remark}{Remark}%

\numberwithin{theorem}{section}
\numberwithin{remark}{section}
\numberwithin{definition}{section}

\usepackage{authblk}

\title{\textbf{Two-Dimensional $\beta$-plane Turbulence: Dual Cascade and Zonal Jets}}

\author[1]{Yuri Cacchiò}
\author[2]{Amirali Hannani}
\author[3]{Gigliola Staffilani}

\affil[1]{Faculty of Mathematics, University of Vienna, Oskar-Morgenstern-Platz 1, 1090 Vienna, Austria}
\affil[2]{Section of Mathematics, University of Geneva, 1205 Geneva, Switzerland.}
\affil[3]{Department of Mathematics, Massachusetts Institute of Technology, 77 Massachusetts Ave, 02139-4307, Cambridge, USA}

\date{\vspace{-5mm}\scriptsize \textbf{E-mails:} \href{mailto:yuri.cacchio@univie.ac.at}{yuri.cacchio@univie.ac.at}, \href{mailto:Amirali.Hannani@unige.ch}{amirali.hannani@unige.ch}, \href{mailto:gigliola@math.mit.edu}{gigliola@math.mit.edu}}

\begin{document}

\maketitle

\begin{abstract}
We derive an exact and novel expression for an averaged two-point correlation function in the statistically stationary, forced-dissipative two-dimensional Navier–Stokes equations subject to the Coriolis force under the $\beta$-plane approximation. This identity is related to the so-called geostrophic balance: it connects the effect of the Coriolis force to the pressure gradient through a two-point correlation function.

Additionally, inspired by \cite{bed2D}, we provide  sufficient conditions under which the asymptotics of the averaged third-order structure function at large spatial scales follow the universal third-order law of two-dimensional turbulence in the absence of the Coriolis force. This complements our previous result \cite{YuriAmiraliGigliola} on small spatial scales.

Together, our results provide a clear picture of the role of the Coriolis force in $\beta$-plane turbulence. On the one hand, the spherically averaged rates of enstrophy and energy transfer are not affected by the Coriolis force. On the other hand, the Coriolis force contributes to anisotropic large-scale organization by altering the spatial distribution of energy and promoting the formation of zonal structures.

The proof relies on a new formulation of the Kármán–Howarth–Monin relation. For the geostrophic balance, we use a novel antisymmetric projection of the KHM relation under which only the pressure and Coriolis terms survive. For the cascade laws, we show that the Coriolis contribution to the averaged classical KHM relation vanishes identically at any scale.

\vspace{1em}
\end{abstract}

\tableofcontents

\section{Introduction}
Geophysical flows such as oceanic currents and atmospheric jets are often described as two-dimensional flows in a rotating frame \cite{pedlosky,vallis1}. A minimal mathematical model for these systems is given by the two-dimensional incompressible Navier-Stokes equations subject to stochastic forcing and dissipation, where planetary rotation is taken into account by a term corresponding to the Coriolis force. 

In the present paper, the velocity field $u(t,x)=(u^1(t,x),u^2(t,x)) \in \mathbb{R}_{\geq 0} \times \mathcal{D}$ and the pressure $p(t,x)$ satisfy the following system
\begin{equation}\label{eq:NS_intro}
\left\{
    \begin{array}{rl}
    \partial_t u + (u \cdot \nabla)u + f(x) u^\perp &= \nu \Delta u - \alpha u - \nabla p + \varphi, \\
    \nabla \cdot u &= 0.
\end{array}
\right.
\end{equation}
Here, $u^\perp = (-u^2, u^1)$, $\nu > 0$ is the kinematic viscosity, $\alpha > 0$ is the large-scale Ekman damping coefficient, and $\varphi$ represents the stochastic forcing. The Coriolis parameter $f(x)$ varies linearly with the meridional coordinate $x^2$ as $f(x) = f_0 + \beta x^2$ under the so-called $\beta$-plane approximation. The constants $f_0$ and $\beta$ are fixed parameters defined by the reference latitude.

The Coriolis parameter intrinsically breaks symmetries, preventing the use of standard doubly-periodic domains \cite{YuriAmiraliGigliola,yuriamiraligigliolaWP, salmon}. Therefore, system \eqref{eq:NS_intro} is posed on a periodic channel $\mathcal{D} = \mathbb{T} \times I$, where $\mathbb{T}=[0, L)$ represents the periodic zonal domain and $I = [a,b]$ is the bounded meridional interval. We impose periodic boundary conditions in the zonal direction ($x^1$) and homogeneous Dirichlet (no-slip) boundary conditions on the meridional boundaries ($x^2 \in \{a,b\}$):
\begin{equation}\label{eq:boundary_conditions}
    u(t, x^1+L, x^2) = u(t, x^1, x^2), \qquad u(t, x^1, a) = u(t, x^1, b) = 0.
\end{equation}
Furthermore, we assume that the stochastic process is spatially regular and white-in-time as in \cite{bed2D, YuriAmiraliGigliola,yuriamiraligigliolaWP}, formulated as
\begin{equation}\label{eq:stochasticprocess}
    \varphi(t,x) = \frac{\partial}{\partial t}\zeta(t,x),\ \ \ \ \zeta(t,x) = \sum_{j=1}^{\infty} b_j \beta_j(t) e_j(x),\ \ \ \ t \geq 0,
\end{equation}
where $\{e_j\}$ is a divergence-free orthonormal basis in $H$ (the completion of divergence-free smooth functions satisfying the boundary conditions in $L^2(\mathcal{D})$) (see Section \ref{sec:setting}), and $\{\beta_j\}$ is a sequence of independent standard Brownian motions.
These Brownian motions are defined on a complete probability space $(\Omega,\mathcal{F},\mathbb{P})$ with a filtration $\mathcal{G}_t$, $t\geq0$, and the $\sigma$-algebras $\mathcal{G}_t$ are completed with respect to $(\mathcal{F},\mathbb{P})$, that is, $\mathcal{G}_t$ contains all $\mathbb{P}$-null sets $A\in \mathcal{F}$. 
The coefficients $b_j$ are constant real numbers such that
\begin{align} 
    \varepsilon &= \frac{1}{2}\sum_{j=1}^\infty b_j^2 < \infty, \label{eq:intro_energy}\\
    \eta &= \frac{1}{2}\sum_{j=1}^\infty b_j^2 \fint_{\mathcal{D}} |\nabla \times e_j|^2 \, dx < \infty,  \label{eq:intro_enstrophyavg}
\end{align}
where $\fint_\mathcal{D}  = \frac{1}{|\mathcal{D}|}\int_\mathcal{D} $ denotes the normalized spatial average over the domain.

The constants $\varepsilon$ and $\eta$ denote the average energy and enstrophy input per unit time per unit area, respectively. Notice that these energy and enstrophy injection rates depend only on the noise and are independent of the viscosity $\nu$ and the drag $\alpha$.

Two-dimensional turbulence exhibits the simultaneous transfer of two conserved quantities: the kinetic energy $E$ and the enstrophy $Z$, defined respectively as
\begin{equation}\label{eq:energy_enstrophy_def}
    E(u) = \frac{1}{2} \fint_{\mathcal{D}} |u|^2 \, dx, \qquad Z(u) = \frac{1}{2} \fint_{\mathcal{D}} |\omega|^2 \, dx,
\end{equation}
where $\omega = \nabla^\perp \cdot u = \partial_1 u^2 - \partial_2 u^1$ is the scalar vorticity.

Throughout the paper, unless otherwise stated, all $L^2$ norms are understood with respect to this normalized spatial average (see \eqref{eq:norms}).

This system exhibits a wide range of physical phenomena such as small-scale eddies, two-dimensional turbulence, and the formation of large-scale organized structures \cite{boffetta, pedlosky, vallis1}. In this paper, we are interested in two of the most prominent phenomena associated with two-dimensional rotating flows: \textit{geostrophic balance} and \textit{two-dimensional turbulence}. In the following, we give a brief overview of these phenomena and then explain our results.

\subsection{Geostrophic balance}
For the system \eqref{eq:NS_intro}, we derive an exact, \textit{novel} statistical balance connecting the Coriolis force to the pressure gradient via two-point correlation functions.

\subsubsection{Definition}
Geostrophic balance states that for large-scale flows in a rotating frame, the Coriolis term dominates the inertial advective terms. This regime occurs when the non-dimensional Rossby number is sufficiently small (we refer the reader to Appendix \ref{app:heuristics} for a detailed dimensional analysis). Consequently, the main balance in the horizontal plane is established between the Coriolis force and the pressure gradient: 
\begin{equation} \label{eq: Geo bal sch 1}
    f(x) u^{\perp} \approx -\nabla p.
\end{equation}
Mathematically, this relation represents the leading-order approximation of the velocity equation \eqref{eq:NS_intro} in the asymptotic limit of large scales.

\subsubsection{Two-point correlation: main contribution}

Here, we are mainly interested in the interplay between the geostrophic balance and the two-point correlation function of the stationary flow. In \cite{YuriAmiraliGigliola} we proved that the model \eqref{eq:NS_intro} admits at least one stationary measure; we denote expectation with respect to this measure by $\mathbb{E}$. Given a separation vector $y \in \mathbb{R}^2$ and $x \in \mathcal{D}$, we want to investigate how the two-point correlation function 
\begin{equation}
    \mathbb{E}(u(x) \cdot u(x+y))
\end{equation}
behaves and whether the geostrophic balance can be observed via this quantity.

We begin with the following heuristic argument: fix a separation vector $y\in \mathbb{R}^2$ with $|y|=l$ and direction $n=y/|y|=(n^1,n^2)$. Assume \eqref{eq: Geo bal sch 1} holds, and consider the quantity $\nabla p(x) \cdot u^{\perp}(x+ln) - u^{\perp}(x) \cdot \nabla p(x+ ln)$. Replacing $\nabla p$ in this expression with $-f(x) u^{\perp}$, using the fact that the shift in the Coriolis parameter is $f(x+ln)-f(x)=\beta l n^2$,and then taking the expectation yields: 
\begin{equation} \label{eq: two point function intro heur}
    \beta l n^2\mathbb{E}(u(x) \cdot u(x+ln)) = \mathbb{E}\big(\nabla p(x) \cdot u^{\perp}(x+ln) - u^{\perp}(x) \cdot \nabla p(x+ ln)\big).
\end{equation}

Our first contribution in this paper is to derive an ``averaged" version of the above identity starting from \eqref{eq:NS_intro}.  
More precisely, we denote the spatial and spherical average by $\langle \cdot \rangle$, i.e., for any function $g(x,n)$, $\langle g(x,n) \rangle = \fint_{\mathcal{D}} \fint_{\mathbb{S}} g(x,n) \, dn \, dx$, where $\fint$ is the normalized integral. Then, we derive and prove the following exact identity from \eqref{eq:NS_intro} (see Theorem \ref{thm:antisymmetric_balance} for the precise statement): 
\begin{equation} \label{eq: balance sch 2}
      \mathbb{E} \langle \beta l (n^2)^2 (u(x) \cdot u(x+ln)) \rangle =  \mathbb{E}\langle n^2 \big(\nabla p(x) \cdot u^{\perp}(x+ln) - u^{\perp}(x) \cdot \nabla p(x+ln) \big)\rangle.
\end{equation}

\subsubsection{Remarks}
First, let us emphasize that identity \eqref{eq: balance sch 2} is a \textit{novel} balance pertaining to the geostrophic regime, derived from first principles (i.e., directly from \eqref{eq:NS_intro}) without any extra assumptions or statistical closures. While exact two-point relations in rotating turbulence do exist, such as KHM equations for potential vorticity fluxes \cite{kurien_PV} or exact energy laws \cite{galtier}, our antisymmetric projection isolates the exact $\beta$-pressure balance. Furthermore, unlike classical relations that assume geostrophic balance \textit{a priori} to link wind and pressure correlations (e.g., \cite{buell}), our identity requires no macroscopic approximations.

In the previous section, we argued how to deduce \eqref{eq: balance sch 2} from the geostrophic balance. However, we do not expect the converse to be strictly true.

To relate in the reverse direction, we refer to Appendix \ref{app:heuristics}. In this appendix, we provide a heuristic scaling argument showing that from \eqref{eq: balance sch 2} one is led to conclude that the flow becomes predominantly zonal, i.e., the velocity is predominantly in the horizontal direction. While the formation of zonal jets is classically studied via one-point zonally averaged equations (such as the Taylor identity \cite{vallis1}) or quasi-linear statistical frameworks (like CE2/S3T \cite{Cost2, Cost}), our identity provides a fully nonlinear, exact two-point mathematical signature of this phenomenon. This remains consistent with the predictions from other physical heuristics, numerical simulations, and observational data \cite{Danilov, galerpin1, rhines, vallis4}.

\subsection{Turbulence}
The investigation of \textit{three-dimensional} stationary forced-dissipative turbulence dates back at least to Kolmogorov’s work in the early 1940s \cite{kolmogorov2,kolmogorov1,kolmogorov3}, where he predicted the direct cascade of energy (i.e., from small to large wavenumbers). Energy injected at large scales by an external force is transferred towards smaller scales, where it is dissipated by viscosity. 

Kolmogorov's argument relies on the fact that only energy is conserved. In \emph{two dimensions}, the phenomenology of statistically stationary, forced-dissipative flows changes drastically due to the simultaneous conservation of two quantities: energy and enstrophy. This additional conservation law leads to the more complicated \emph{dual cascade} picture, first argued by Fjørtoft \cite{Fjrtoft}, and later by Batchelor and Kraichnan \cite{Batchelor2,Kraichnan,Kraichnan2}. 

In this regime, we observe an inverse energy cascade from small spatial scales (large wavenumbers) to large spatial scales (small wavenumbers), and a direct enstrophy cascade from large spatial scales to small spatial scales. 
  
Let us provide a more precise qualitative picture of this dual cascade. We denote the dissipation, energy injection, and friction scales by $l_\nu$, $l_I$, and $l_\alpha$, respectively, with the scale separation $l_\nu \ll l_I \ll l_\alpha$. Here, $l_\nu$ is the scale where viscosity predominantly acts, $l_I$ is the scale where energy is injected, and $l_\alpha$ is the scale where the large-scale dissipation (drag) is dominant. The dual cascade picture predicts that energy is transferred from the injection scale to the friction scale and damped by friction (\emph{inverse cascade}). On the other hand, enstrophy cascades from the injection scale to the dissipation scale and is dissipated by viscosity (\emph{direct cascade}).
 
In this classical picture, the Coriolis force is absent. One of the main purposes of the present paper is to understand which parts of this phenomenology survive in the presence of the \(\beta\)-plane Coriolis force.

\subsubsection{Quantitative Characterization}

A quantitative characterization of the dual cascade is provided via either of the following quantities \cite{frisch}:
\begin{enumerate}
    \item The ensemble average of the energy spectrum $E(k)$:
\begin{equation}
    E(k) \sim |k|\mathbb{E} (|\hat{u}(k)|^2),
\end{equation}
where $\hat{u}(k)$ denotes the velocity Fourier coefficient, and the expectation is taken with respect to the stationary measure.

\item The $n$-th order structure functions: For any separation vector $y \in \mathbb{R}^2$, let $\delta_{y}u(x):=u(x+y)-u(x)$. The $n$-th order structure function is defined as
\begin{equation}
    S_n(y) := \mathbb{E}(|\delta_yu(x)|^n),
\end{equation}
where the dependence on $x$ is suppressed on the left-hand side (typically by assuming spatial homogeneity). 
\end{enumerate}

Assuming isotropy and homogeneity in three dimensions, the celebrated Kolmogorov 4/5-law (the quantitative formulation of the direct cascade) predicts that the third-order longitudinal structure function behaves as
\begin{equation}\label{thirdorderenergy}
    S_3^{\parallel}(l) := \mathbb{E}\left[\left(\delta_y u \cdot \frac{y}{|y|}\right)^3\right] \sim -\frac{4}{5} \varepsilon l, \quad l_\nu \ll l \ll l_I,
\end{equation}
where $l = |y|$ is the separation scale, $l_\nu$ is the dissipation scale, $l_I$ is the energy injection scale, and $\varepsilon$ denotes the average energy injection rate.

In two dimensions, at \textit{small scales}, the direct cascade of enstrophy can be characterized by the following mixed velocity-vorticity structure function as predicted by Eyink \cite{eyink},
\begin{equation}\label{eq: velocity-vorticity str}
    S_3^{(1)}(l)=\mathbb{E}\left(|\delta_{ln} \omega|^2 (\delta_{ln} u\cdot n) \right)\sim -2\eta l, 
 \ \ \ l_\nu \ll l \ll l_I,
\end{equation}
which was originally derived in \cite{yaglom} for passive scalar turbulence. Here, $\delta_{ln}\omega:= \omega(x+ln)-\omega(x)$ for a separation vector $y=ln$, with $n=y/|y| \in \mathbb{S}$. This law describes the scale-space flux of enstrophy, where the negative sign indicates a downscale (direct) transfer. 

Yet another way to look at the flux at small scales is by measuring the total velocity-increment energy \cite{bed2D},
\begin{equation}\label{eq: velocity-velocity small scale}
    S_3^{(2)}(l)=\mathbb{E}\left(|\delta_{ln} u|^2 (\delta_{ln} u\cdot n)\right)\sim \frac{1}{4}\eta l^3, 
 \ \ \ l_\nu \ll l \ll l_I.
\end{equation}
At \textit{large scales}, the inverse cascade of energy is characterized by \cite{bed2D},
\begin{equation}\label{eq: velocity-velocity large scale}
    S_3^{(3)}(l)=\mathbb{E}\left(|\delta_{ln} u|^2 (\delta_{ln} u\cdot n)\right)\sim 2\varepsilon l, 
 \ \ \ l_I \ll l \ll l_{\alpha},
\end{equation}
where $S^{(3)}_3$ is the flux of kinetic energy through the separation scale $l$.

We note that the longitudinal versions of these relations can be considered as the two-dimensional analogues of the 4/5-law \eqref{thirdorderenergy}. They were conjectured only in 1999 by Bernard, Lindborg, and Yakhot \cite{bernard,lindborg,yakhot},
\begin{subequations}
\begin{align}
S_3^{\parallel}(l)&=\mathbb{E}\left[\left(\delta_{ln} u\cdot n \right)^3\right]\sim \frac{1}{8}\eta l^3, \ \ \
 l_\nu \ll l \ll l_I, \label{thirdorederenstrophy2d}\\
S_3^{\parallel}(l)&=\mathbb{E}\left[\left(\delta_{ln} u\cdot n \right)^3\right]\sim \frac{3}{2}\varepsilon l,  \ \ \ l_I \ll l \ll l_\alpha, \label{thirdorederenergy2d}
\end{align}
\end{subequations}
where, as before, $\varepsilon$ and $\eta$ are the energy and enstrophy injection rates.

The spectral counterpart of the dual cascade picture has also been predicted over the aforementioned scales \cite{Batchelor2,Kraichnan},
\begin{subequations}
\begin{align}
    |k|&\mathbb{E}|\hat{u}(k)|^2\sim\varepsilon^{2/3}|k|^{-5/3},\ \ \ l^{-1}_{\alpha}\ll|k|\ll l^{-1}_I,\label{energyspectrum}\\
    |k|&\mathbb{E}|\hat{u}(k)|^2\sim\eta^{2/3}|k|^{-3},\ \ \ \ l^{-1}_{I}\ll|k|\ll l^{-1}_\nu.\label{enstrophyspectrum}
\end{align}
\end{subequations}
While this picture is well understood physically and numerically \cite{boffetta,vallis1}, mathematically rigorous proofs have been developed only recently \cite{bed3D,bed2D,dudley,pappa,papathanasiou,buler} (see also \cite{Novak} and references therein for the deterministic counterpart).  

Among the mathematical results above, let us highlight the work of Bedrossian et al. \cite{bed2D}. In this work, the authors introduce a weak sufficient condition, referred to as Weak Anomalous Dissipation (WAD) (see equations \eqref{eq:WAD_energy} and \eqref{eq:WAD_enstrophy} for the precise statement), to rigorously prove an averaged version of \eqref{eq: velocity-vorticity str}, \eqref{eq: velocity-velocity small scale}, and \eqref{eq: velocity-velocity large scale}. By ``averaged version," we mean that instead of assuming spatial homogeneity and isotropy, they consider the spatially and spherically averaged third-order structure functions: 
\begin{equation}
    \overline{S}^{(i)}_3  :=  \langle S^{(i)}_3 \rangle, 
\end{equation}
for $i \in \{1,2,3\}$, where we recall the definitions of $S_3^{(i)}$ from above, and the averaging operator is defined as $\langle g \rangle := \fint_{\mathbb{S}}  \fint_{\mathcal{D}} g(x,n) \, dx \, dn$.

\subsubsection{Effect of the Coriolis force}

So far, we have not considered the effect of the Coriolis force. In the physics literature, the cascade behavior in the presence of the Coriolis force is less understood. We are not aware of any conjecture pertaining to the third-order structure functions in the presence of the Coriolis force. However, regarding the spectral behavior, it is conjectured \cite{chekhlov, Cost2, Cope, galerpin1, huang2, rhines, vallis1} 
 that the presence of the Coriolis force modifies the spectral distribution at large scales, dividing the spectrum into three distinct ranges
\begin{subequations}
\begin{align}
    |k|&\mathbb{E}|\hat{u}(k)|^2 \sim \beta^2|k|^{-5}, \quad\quad l^{-1}_{\alpha} \ll |k| \ll l^{-1}_{R}, \label{energyspectrumCor_zonal}\\
    |k|&\mathbb{E}|\hat{u}(k)|^2 \sim \varepsilon^{2/3}|k|^{-5/3}, \quad l^{-1}_{R} \ll |k| \ll l^{-1}_I, \label{energyspectrumCor_iso}\\
    |k|&\mathbb{E}|\hat{u}(k)|^2 \sim \eta^{2/3}|k|^{-3}, \quad\quad l^{-1}_{I} \ll |k| \ll l^{-1}_\nu. \label{enstrophyspectrumCor}
\end{align}
\end{subequations}
Here, $l_{R}$ represents the transition scale at which the inverse energy cascade first feels the $\beta$-effect (see \eqref{eq:balance_beta_eps} for the  derivation in our framework). Notice that, unlike in the standard isotropic turbulence picture presented in the previous section, the friction scale $l_{\alpha}=l_{\alpha}(\beta)$ is also modified by the planetary rotation.

Comparing the above relations with \eqref{energyspectrum} and \eqref{enstrophyspectrum} paints a clear picture: the Coriolis force does not alter the qualitative behavior of the dual cascade at small and intermediate scales. Indeed, both the direct enstrophy cascade and the initial $-5/3$ inverse energy cascade remain completely unaffected. It is only at the largest scales ($|k| \ll l_R^{-1}$) that the Coriolis force modifies the energy distribution.

In \cite{YuriAmiraliGigliola}, we proved a result confirming the above picture at small scales. Our result pertains to the third-order structure functions: we provided a sufficient condition under which the behavior of the averaged third-order structure functions \eqref{eq: velocity-vorticity str} and \eqref{eq: velocity-velocity small scale} in the presence of the Coriolis force remains identical to the case where the Coriolis force is absent.

\subsubsection{Main contribution}

We prove, assuming the weak anomalous dissipation (WAD) hypothesis, that the behavior of the averaged third-order structure function $\overline{S}^{(3)}_3$ remains unaffected by the Coriolis force. The precise statement of the result is given in Theorem \ref{thm:inverse_cascade}. We give an informal version here.

Consider equation \eqref{eq:NS_intro}, and denote the expectation with respect to the stationary measure by $\mathbb{E}$. The existence of such a measure is proved in \cite{yuriamiraligigliolaWP}. 
The weak anomalous dissipation hypothesis is given in \eqref{eq:WAD_energy} and \eqref{eq:WAD_enstrophy}: physically, it assumes that all the injected energy will be dissipated by the large-scale friction term and all the injected enstrophy will be dissipated by the viscosity, uniformly in $\nu, \alpha$:
\begin{equation*}
    \lim_{\nu \to 0} \sup_{\alpha \in (0,1)} \alpha \mathbb{E}\|u_{\nu,\alpha} \|^2_{L^2(\mathcal{D})}=\varepsilon, \quad  \lim_{\alpha \to 0} \sup_{\nu\in(0,1)} \nu \mathbb{E}\|\nabla \omega_{\nu,\alpha} \|^2_{L^2(\mathcal{D})}=\eta.
\end{equation*}

Then, our main result, Theorem \ref{thm:inverse_cascade}, states that at large scales ($l_I \ll l \ll l_{\alpha}$) we have
\begin{equation} \label{third order def intro}
    \frac{1}{|\mathcal{D}_\alpha|}\,\mathbb{E} \fint_{\mathbb{S}}
    \int_{\mathbb{T}_{L(\alpha)} \times \mathbb{R}}
    |\delta_{ln}u(x)|^2 (\delta_{ln}u(x) \cdot n) \, dx \, dn \sim 2 \varepsilon l,
\end{equation}
where $\varepsilon$ is the average injected energy.

The above asymptotic should be understood for large scales $l_{I}\ll l \ll l_{\alpha}$ in the limit $\alpha,\nu \to 0$. The notion of \textit{large scale} should be understood in the following limiting regime, similar to \cite{bed2D}: we fix a separation scale $l \in [l_I, l_{\alpha}]$, and to observe the large scales ($l \to \infty$) we take the limit $\nu,\alpha \to 0$ such that $l_{\alpha} \to \infty$ at an appropriate rate. Finally, we take the limit $l_{I} \to \infty$.

Notice that, since the transition scale $l_R \sim (\varepsilon/\beta^3)^{1/5}$ (see \eqref{eq:balance_beta_eps}) is independent of the friction parameter $\alpha$, taking the asymptotic limit $l_\alpha \to \infty$ (i.e. $\alpha\to 0$) implies that the inertial range $[l_I, l_\alpha]$ crosses the transition scale $l_R$. Therefore, Theorem \ref{thm:inverse_cascade} establishes that the spherically averaged energy flux remains completely unaffected by the planetary rotation even inside the zonostrophic regime ($l > l_R$).

Let us emphasize that in this regime, since the size of the separation vector grows to infinity, the domain size must also grow to infinity at an appropriate rate (see Section \ref{sec:main_cascade}). 

We also refer to Section \ref{sec: boundary layer} for a detailed discussion on the definition of \ref{third order def intro}; in particular, in Section \ref{rem:quadratic_vs_cubic} we explain the different choice of ``averaging" domains and their relevance.

Besides the above results, we also improve our previous result from \cite{YuriAmiraliGigliola}. In \cite{YuriAmiraliGigliola}, we proved the enstrophy cascade at small scales provided that $\mathbb{E}\|\omega \|_{L^2}^2$ remains bounded uniformly in $\alpha, \nu>0$. This assumption is stronger than WAD. In this paper, we recover the result of \cite{YuriAmiraliGigliola}, i.e., the averaged versions of \eqref{eq: velocity-vorticity str} and \eqref{eq: velocity-velocity small scale}, assuming WAD.

The precise statement is given in Theorem \ref{thm:enstrophy_cascade}.

\subsubsection{Interpretation of the result}

At first sight, our result seems consistent with the spectral prediction at intermediate scale \eqref{energyspectrumCor_iso} but inconsistent with the spectral prediction at large scales \eqref{energyspectrumCor_zonal}. Let us explain why this is not the case: our result indicates that the averaged \textit{isotropic} and \textit{homogeneous} rate of energy transfer between different spatial scales is not affected by the Coriolis force. This is somewhat expected since the Coriolis force does not perform any local work $(f u^{\perp} \cdot u=0)$.  

The main effect of the Coriolis force is to create zonal structures and reorganize the spatial distribution of the energy, without changing the net average rate of energy transfer. In fact, our result on the geostrophic balance \eqref{eq: balance sch 2} provides a rigorous representation of this phenomenon. 

Let us give a physical heuristic on why \eqref{eq: balance sch 2} confirms the above picture. For an isotropic flow, one can verify that the right-hand side of \eqref{eq: balance sch 2} (the pressure gradient correlation) becomes identically zero for any $l \ge 0$. Furthermore, for a flow on a standard doubly-periodic domain in the absence of the $\beta$-effect, the right-hand side of \eqref{eq: balance sch 2} is also zero, independent of the isotropy assumption. 

One might therefore suspect that, even in the presence of the Coriolis force, this term vanishes and \eqref{eq: balance sch 2} trivially reduces to $0=0$. Let us argue mathematically why this is not the case. If we fix $\nu > 0$ and take the limit $\alpha \to 0$, which is the relevant regime to isolate large-scale dynamics (see, e.g., Theorem 1.21 of \cite{bed2D}), the energy and enstrophy balances \eqref{eq:energy_balance}-\eqref{eq:enstrophy_balance} imply that $\mathbb{E}\|u\|_{L^2}^2$ scales as $\varepsilon/\alpha$. Consequently, the left-hand side of \eqref{eq: balance sch 2} behaves asymptotically as $\varepsilon \beta l / \alpha$ (see Appendix \ref{appB} for further details). This demonstrates that the term is strictly positive and far from a trivial identity.

Moreover, from a physical perspective, as the separation scale $l$ increases and approaches the large-scale inertial range, the left-hand side of \eqref{eq: balance sch 2} grows due to the $\beta l$ weight. As we heuristically argue via dimensional scaling in Appendix \ref{app:heuristics}, the unbounded growth of the $\beta$-term must be compensated by the pressure correlation. This compensation forces a structural reorganization of the flow, providing a strong mathematical signature of how the Coriolis force induces the anisotropic zonal structures.

If one is interested in observing the breaking of isotropy due to the Coriolis force directly within the third-order structure functions, the correct observable is the longitudinal third-order structure function $S_3^{\parallel}$ introduced in \eqref{thirdorederenstrophy2d}-\eqref{thirdorederenergy2d}. We expect the behavior of $S_3^{\parallel}$ to change drastically in the presence of the Coriolis force. However, we are currently unable to prove this. In the next section, we provide a general argument illustrating why deriving such statements is out of reach for current methods.

\subsection{Overview of the Proofs}

The first important ingredients of our proofs are the energy and enstrophy balances, which are given in \eqref{eq:energy_balance} and \eqref{eq:enstrophy_balance}. They can be derived directly from \eqref{eq:NS_intro} using Itô's formula. Notably, these balances are identical to the case where the Coriolis force is absent \cite{YuriAmiraliGigliola}. 

Our strategy, following the framework of \cite{bed2D, YuriAmiraliGigliola}, is based on the Kármán-Howarth-Monin (KHM) relations \cite{karman, monin}. Let us briefly explain this approach: for two points $x, x+y \in \mathcal{D}$, we consider the tensor product $u(x) \otimes u(x+y)$. One can compute the time derivative of this quantity based on Itô's formula and the evolution equation \eqref{eq:NS_intro}. Afterwards, taking the expectation of the resulting identity with respect to the stationary measure eliminates the time derivative. This procedure yields an exact statistical identity involving two-by-two matrices. Taking the trace of this resulting tensor equation recovers the standard scalar KHM relation. 

We first explain the proofs of Theorems \ref{thm:enstrophy_cascade} and \ref{thm:inverse_cascade}. The key mechanism is that the expected value of the advective term $(u \cdot \nabla)u$, after simplifications, yields the third-order structure functions. The subsequent step relies on extracting the correct asymptotics of all remaining terms using the stationary balances alongside the WAD hypothesis. These steps remain structurally identical to the non-rotating case. In our previous work \cite{YuriAmiraliGigliola}, we showed that the new trace contribution arising from the Coriolis force is asymptotically small. However, in the present work, we prove that this term, due to the incompressibility and the integration over the domain and the unit sphere, is identically zero (see Section \ref{sec:vanishing}). This provides a definitive indication that the Coriolis force does not affect the \textit{average} isotropic rate of the dual cascade at any scale. 

To formally extract the geostrophic balance and prove Theorem \ref{thm:antisymmetric_balance}, we reconsider the aforementioned matrix equation. Instead of taking the trace, we project the tensor using a weighted antisymmetric sum of its off-diagonal components (see the exact test function in Section \ref{sec:khm_anti}). Such a combination transforms the Coriolis contribution into the weighted two-point function present on the left-hand side of \eqref{eq: balance sch 2}. Moreover, all other terms, except the one involving the pressure gradients, cancel out under this antisymmetric projection. In the standard trace formulation, the pressure term vanishes due to spatial symmetries and incompressibility. Here, it survives and provides the main balancing mechanism. 

Finally, let us comment on why obtaining a generalized law for the longitudinal third-order structure function is highly non-trivial. If one projects the tensor equation to extract the longitudinal flux, the corresponding Coriolis contribution does not vanish. Attempting to control this remaining Coriolis term leads to a complicated sign structure and introduces bounds that diverge in the inviscid limit under standard WAD assumptions (as discussed in Section \ref{sec:main_cascade}), thereby preventing a rigorous asymptotic closure.

\paragraph{Organization of the paper.}
In Section \ref{sec:setting}, we establish the functional framework, recall the well-posedness and regularity of the stationary measure for the stochastic Navier-Stokes equations on the $\beta$-plane, and we introduce the weak Kármán-Howarth-Monin identities. Section \ref{sec:main_results} presents the statements of our main results: the exact cancellation of the spherically averaged Coriolis terms (Section \ref{sec:cancellation_main_r}), the rigorous derivation of the dual cascade under the Weak Anomalous Dissipation framework (Section \ref{sec:main_cascade}), and the novel antisymmetric KHM balance governing the anisotropic dynamics (Section \ref{sec:khm_anti}). Section \ref{sec:proofs} is devoted to the proofs of these theorems. Finally, Appendices \ref{app:heuristics} and \ref{appB} provide a discussion connecting our exact antisymmetric identity to classical geophysical heuristics.

\section{Mathematical Setting and Preliminaries}\label{sec:setting}

In this section, we establish the main framework for the stochastic two-dimensional Navier-Stokes equations under the $\beta$-plane approximation. Moreover, we recall the well-posedness and regularity results for the stationary measure obtained in our previous work \cite{yuriamiraligigliolaWP} (see also \cite{mustafa, Yuri, gallagher2, gallagher1} for further rigorous mathematical results on these equations).

Following the notation introduced in \cite{yuriamiraligigliolaWP}, let $\mathcal{V}$ be the space of smooth, divergence-free test functions
\begin{equation*}
    \mathcal{V} = \big\{ \phi \in C^\infty(\bar{\mathcal{D}}, \mathbb{R}^2) : \nabla \cdot \phi = 0, \ \phi \text{ is periodic in } x^1, \ \phi = 0 \text{ on } \partial\mathcal{D} \big\}.
\end{equation*}
We define the following Hilbert spaces for incompressible fluids as the closures of $\mathcal{V}$ in $L^2(\mathcal{D}, \mathbb{R}^2)$ and $H^1_0(\mathcal{D}, \mathbb{R}^2)$ respectively,
\begin{align}
    H &= \overline{\mathcal{V}}^{L^2(\mathcal{D}, \mathbb{R}^2)}, \\
    V &= \overline{\mathcal{V}}^{H^1_0(\mathcal{D}, \mathbb{R}^2)}.
\end{align}
These spaces are equipped with the normalized norms 
\begin{equation}\label{eq:norms}
    \|u\|_H = \|u\|_{L^2(\mathcal{D})}, \qquad \|u\|_V = \|u\|_{H^1(\mathcal{D})},
\end{equation}
where
\begin{align}
    \|u\|_{L^2(\mathcal{D})}^2&=
    \fint_\mathcal{D} |u(x)|^2\,dx=
    \frac{1}{|\mathcal{D}|}\int_\mathcal{D} |u(x)|^2\,dx,\\
     \|u\|_{H^1(\mathcal{D})}^2&=\fint_{\mathcal D} |u(x)|^2\,dx+\fint_{\mathcal D} |\nabla u(x)|^2\,dx.
\end{align}
Let $P_H : L^2(\mathcal{D}, \mathbb{R}^2) \to H$ denote the Leray-Helmholtz orthogonal projection \cite{Leray1,Leray3,Leray2}. By applying $P_H$ to the governing system \eqref{eq:NS_intro}, we eliminate the pressure gradient and obtain the abstract stochastic evolution equation in $H$,
\begin{equation}\label{eq:abstract_NS}
    du + \big[ \nu A u + B(u,u) + \alpha u + \mathcal{C}(u) \big] dt = d\zeta(t),
\end{equation}
where $A = -P_H \Delta$ is the Stokes operator, $B(u,u) = P_H((u \cdot \nabla)u)$ is the bilinear inertial operator, $\mathcal{C}(u) = P_H(f(x)u^\perp)$ is the Coriolis operator, and $d\zeta(t) = P_H(\varphi) dt$ is the projected stochastic forcing. It is well-known that solving the projected equation \eqref{eq:abstract_NS} in $H$ is equivalent to solving the original system \eqref{eq:NS_intro} for the velocity field (see \cite{yuriamiraligigliolaWP} for more details).

\subsection{Well-Posedness and Stationary Measure}
The inclusion of the Coriolis force on a bounded channel breaks the standard symmetries, posing additional challenges for well-posedness and regularity.  

In what follows, we consider solutions to \eqref{eq:abstract_NS} in the weak sense. We refer the reader to \cite{yuriamiraligigliolaWP} for the precise definition of the stochastic solution and the filtered probability space. We now recall the existence results.

\begin{theorem}[Well-Posedness and Stationary Measure \cite{yuriamiraligigliolaWP}]\label{thm:stationary_measure}
For any $\nu > 0$, $\alpha > 0$, and initial datum $u_0 \in H$, the stochastic $\beta$-plane Navier-Stokes equation \eqref{eq:abstract_NS} admits a unique global-in-time solution $u \in C(\mathbb{R}_+; H) \cap L^2_{loc}(\mathbb{R}_+; V)$ in the weak sense, almost surely. Furthermore, the Markov semigroup associated with the solutions admits at least one invariant probability measure $\mu_{\nu,\alpha}$ supported on $V$.
\end{theorem}
Hereafter, $\mathbb{E}[\cdot]$ denotes the expectation with respect to the stationary measure $\mu_{\nu,\alpha}$. While Theorem \ref{thm:stationary_measure} guarantees the existence of a stationary measure, the rigorous derivation of the Kármán-Howarth-Monin (KHM) relation requires higher regularity. By assuming $\varepsilon_1 := \sum_j \lambda_j b_j^2 < \infty$, where $\lambda_j$ are the eigenvalues of the Stokes operator $A$, the support of the stationary measure has $H^2$ regularity.

\begin{theorem}[$H^2$ Regularity \cite{yuriamiraligigliolaWP}]\label{thm:H2_regularity}
Under the assumption $\varepsilon_1 < \infty$, the flow satisfies a uniform stationary enstrophy balance, yielding the following bound with respect to the invariant measure $\mu_{\nu,\alpha}$:
\begin{equation}\label{eq:H2_bound}
    \nu \mathbb{E} \|A u\|_{H}^2 \le \varepsilon_1.
\end{equation}
Consequently, $\mu_{\nu,\alpha}$ is supported on the domain of the Stokes operator $D(A) = H^2(\mathcal{D}) \cap V$, and the stationary velocity fields are almost surely in $H^2(\mathcal{D})$.
\end{theorem}
This $H^2$ regularity guarantees that the two-point correlation functions in the KHM relations are uniformly bounded and at least $C^2$ with respect to the spatial separation vector $l$.

\subsection{Weak Kármán-Howarth-Monin Relation}\label{sec:weak_khm}

{For $x \in \mathcal{D}$ and a separation vector $y \in \mathbb{R}^2$, the shifted point $x+y$ need not lie in $\mathcal{D}$, so that the two-point
quantities below are a priori undefined. Moreover, unlike the doubly periodic
case, meridional shifts do not preserve the domain, $\mathcal{D} - y \neq
\mathcal{D}$. As derived in \cite{YuriAmiraliGigliola}, we therefore extend $u$, $\nabla p$ and the noise by zero to $\mathbb{T}\times\mathbb{R}$. Then,  the correlations below become well-defined
functions of $y$ on the whole of $\mathbb{R}^2$, and the used manipulations in $y$ below (symmetrization, integration by parts, testing against
$\phi \in C_c^\infty(\mathbb{R}^2)$) are legitimate.

To keep the notation simple, we use the same symbol for the extended functions
when it does not create confusion.}

For any separation vector $y \in \mathbb{R}^2$, we define the following  averaged quantities: second-order velocity correlation tensor $\Gamma(y)$, the third-order structure tensor $D^j(y)$ (for $j=1,2$), the Coriolis correlation tensor $\varTheta(y)$, and the noise correlation tensor $a(y)$, respectively, as
\begin{align}
    \Gamma(y) &= \mathbb{E} \fint_{\mathbb{T} \times \mathbb{R}} u(x) \otimes u(x+y) \, dx, \label{def:Gamma}\\
    D^j(y) &= \mathbb{E} \fint_{\mathbb{T} \times \mathbb{R}} (\delta_y u(x) \otimes \delta_y u(x)) \delta_y u^j(x) \, dx, \label{def:Dj}\\
    \varTheta(y) &= \frac{1}{2}\mathbb{E} \fint_{\mathbb{T} \times \mathbb{R}} \big( u(x) \otimes (f u^\perp)(x+y) + (f u^\perp)(x) \otimes u(x+y) \big) \, dx, \label{def:Theta}\\
    a(y) &= \frac{1}{2}\sum_{j} b_j^2 \fint_{\mathbb{T}\times \mathbb{R}} e_j(x) \otimes e_j(x+y) \, dx, \label{def:a}
\end{align}
where $\delta_y u(x) := u(x+y) - u(x)$ is the velocity increment, $\{e_j\}$ is a divergence-free orthonormal basis in $H$ introduced in \eqref{eq:stochasticprocess}, and $\fint$ denotes the spatial integral over $\mathbb{T} \times \mathbb{R}$ normalized by the area of the bounded domain $|\mathcal{D}|$. Similarly, we define the pressure cross-correlation tensor as
\begin{equation}\label{def:Pi}
    \Pi(y) = \frac{1}{2}\mathbb{E} \fint_{\mathbb{T} \times \mathbb{R}} \big[ \nabla p(x) \otimes u(x+y) + u(x) \otimes \nabla p(x+y) \big] \, dx.
\end{equation}

Applying Itô's formula to the tensor product $u(x) \otimes u(x+y)$ and taking the expectation with respect to the stationary measure $\mu_{\nu,\alpha}$ yields the following identity \cite[eq. (3.19)]{YuriAmiraliGigliola}.

\begin{proposition}[Weak KHM Equation \cite{YuriAmiraliGigliola}] \label{prop:weak_KHM}
For any smooth matrix-valued test function $\phi \in C_c^\infty(\mathbb{R}^2, \mathbb{R}^{2\times 2})$, the statistically stationary flow $u$ satisfies the identity
\begin{align}\label{eq:Master_KHM}
    &-\sum_{j=1}^2 \int_{\mathbb{R}^2} \partial_{y^j} \phi(y) : D^j(y) \, dy + 4\nu \int_{\mathbb{R}^2} \Delta \phi(y) : \Gamma(y) \, dy - 4\alpha \int_{\mathbb{R}^2} \phi(y) : \Gamma(y) \, dy \notag \\
    &- 4\int_{\mathbb{R}^2} \phi(y) : \varTheta(y) \, dy - 4\int_{\mathbb{R}^2} \phi(y) : \Pi(y) \, dy + 4\int_{\mathbb{R}^2} \phi(y) : a(y) \, dy = 0,
\end{align}
where $A : B = \text{Tr}(A^T B)$ denotes the Frobenius inner product.
\end{proposition}
Similarly, we derive equation \eqref{eq:Master_KHM} in terms of vorticity $\omega = \nabla^\perp \cdot u$. 
We define the scalar two-point spatial correlation functions for the vorticity $\mathfrak{C}(y)$, the Coriolis term $\mathfrak{Q}(y)$, and the noise $\mathfrak{a}(y)$, respectively, as follows:
\begin{align}
    \mathfrak{C}(y) &= \mathbb{E}\fint_{\mathbb{T}\times \mathbb{R}} \omega(x)\omega(x+y) \, dx, \label{def:vorticity_corr}\\
    \mathfrak{Q}(y) &= \frac{1}{2}\beta \, \mathbb{E}\fint_{\mathbb{T}\times \mathbb{R}} \big( u^2(x)\omega(x+y) + u^2(x+y)\omega(x) \big) \, dx, \label{def:coriolis_vort_corr}\\
    \mathfrak{a}(y) &= \frac{1}{2}\sum_{j=1}^\infty b_j^2\fint_{\mathbb{T}\times \mathbb{R}} \big(\nabla^\perp \cdot e_j(x)\big)\big(\nabla^\perp \cdot e_j(x+y)\big) \, dx. \label{def:noise_curl_corr}
\end{align}
The corresponding third-order enstrophy flux vector $\mathfrak{D}(y)$ is defined as
\begin{equation}\label{def:D_vorticity}
    \mathfrak{D}(y) = \mathbb{E}\fint_{\mathbb{T}\times \mathbb{R}} |\delta_y\omega(x)|^2\delta_y u(x) \, dx.
\end{equation}

\begin{proposition}[Vorticity KHM Relation \cite{YuriAmiraliGigliola}]\label{prop:vorticity_KHM}
Let $\omega$ be the vorticity of the statistically stationary flow. Then, for any separation vector $y \in \mathbb{R}^2$, the following enstrophy balance holds 
\begin{equation}\label{eq:vorticity_khm}
    \nabla_y \cdot \mathfrak{D}(y) = -4\nu\Delta_y\mathfrak{C}(y) + 4\alpha\mathfrak{C}(y) + 4\mathfrak{Q}(y) - 4\mathfrak{a}(y).
\end{equation}
\end{proposition}

\subsubsection{Remark on the choice of averaging} \label{rem:quadratic_vs_cubic}
In the definition of $\Gamma,D^j,\varTheta,a,\Pi$ as well as $\mathfrak{C},\mathfrak{Q},\mathfrak{a}, \mathfrak{D}$, to obtain a quantity that only depends on the separation vector $y$, we averaged over  $\mathbb{T} \times \mathbb{R}$. Another natural choice would be to average over the domain $\mathcal{D}$. For example  alternatively we could have defined $\Gamma(y):= \mathbb{E}\fint_{\mathcal{D}} u(x) \otimes u(x+y) dx$.   \\

First, notice that the quantities $\Gamma,\varTheta,a,\Pi $ and $\mathfrak{C},\mathfrak{Q},\mathfrak{a}$ are insensitive to such a choice of the domain. \textit{Consequently,  from now onward we may switch the domain of integration  between $\mathcal{D}$ and $\mathbb{T} \times \mathbb{R}$ from time to time without explicitly mentioning, whenever we deal with these quantities ($\Gamma,\varTheta,a$ and $\mathfrak{C},\mathfrak{Q}$,$\mathfrak{a}$).} \\

On the other hand, $D^j$ and $\mathfrak{D}$ depends on the choice of the integration domain. We discuss the resulting discrepancy in Section \ref{sec: boundary layer}.

\section{Main Results}\label{sec:main_results}

In this section, we state the main results of the paper. We first establish the cancellation of the Coriolis fluxes within the spherically averaged Kármán-Howarth-Monin framework. This averaging over the unit sphere $\mathbb{S}$ is required to write the identity \eqref{eq:Master_KHM} in terms of the scalar separation scale $l = |y|$, recovering the standard cascade laws of classical isotropic turbulence \eqref{eq: velocity-vorticity str},  \eqref{thirdorederenstrophy2d}-\eqref{thirdorederenergy2d}. Thereafter, we state the novel antisymmetric KHM balance governing the large-scale zonal anisotropy.
\subsection{Vanishing of Coriolis Terms}\label{sec:cancellation_main_r}
To estimate the effects of planetary rotation on the turbulent cascades, we define the spherically averaged Coriolis velocity correlation $\overline{\varTheta}(l)$ and the spherically averaged enstrophy Coriolis correlation $\overline{\mathfrak{Q}}(l)$ as
\begin{align}
    \overline{\varTheta}(l) &:= \fint_{\mathbb{S}} \text{Tr}\big(\varTheta(ln)\big) \, dn, \label{eq:Theta_trace_main}\\
    \overline{\mathfrak{Q}}(l) &:= \fint_{\mathbb{S}} \mathfrak{Q}(ln) \, dn, \label{eq:Q_corr_main}
\end{align}
where $\varTheta(y)$ and $\mathfrak{Q}(y)$ are the two-point correlations defined in \eqref{def:Theta} and \eqref{def:coriolis_vort_corr}, respectively. These terms represent the contribution of the rotation to the net energy and enstrophy transfer rates. We prove that these contributions are identically zero.

\begin{theorem}[Vanishing of Coriolis Terms]\label{thm:vanishing_coriolis}
Let $u$ be a statistically stationary solution to the $\beta$-plane stochastic Navier-Stokes equations \eqref{eq:abstract_NS}. Then, the spherically averaged Coriolis contributions vanish for all scales $l > 0$:
\begin{equation}
    \overline{\varTheta}(l) = 0 \quad \text{and} \quad \overline{\mathfrak{Q}}(l) = 0.
\end{equation}
\end{theorem}
As we demonstrate in Section \ref{sec:vanishing}, this result arises from a stronger pointwise identity. We prove that $\text{Tr}(\varTheta(y)) = 0$ and $\mathfrak{Q}(y) = 0$ for every separation vector $y \in \mathbb{R}^2$, implying that the Coriolis force does not contribute to the isotropic flux balances. Planetary rotation neither injects nor dissipates energy or enstrophy, acting as a geometric constraint.

\subsection{Dual Cascade under Weak Anomalous Dissipation}\label{sec:main_cascade}

The goal of this section is to present two theorems showing that the Coriolis force does not affect the behavior of certain ``averaged" third-order structure functions. 
The intuition is clear from the previous section's observation that the Coriolis terms $\overline{\varTheta} = 0$ and $\overline{\mathfrak{Q}} = 0$ vanish and the new spherically averaged KHM reduces to that of the non-rotating fluid. 
 
This cancellation allows us to recover the result of \cite{YuriAmiraliGigliola} with a weaker assumption. In \cite{YuriAmiraliGigliola} we proved certain asymptotics for the ``third order structure" function at small scales, assuming that the expectation of the vorticity is uniformly bounded in $\alpha,\nu$. Here instead we prove the same result assuming the so-called weak anomalous dissipation (WAD) introduced by \cite{bed2D} which is weaker than assuming that the vorticity is uniformly bounded.

Since the dual cascade requires an unbounded inertial range, we work in a large-domain limit. We define the expanding sequence of channels
$\mathcal{D}_\alpha = \mathbb{T}_{L(\alpha)} \times [a(\alpha), b(\alpha)]$,
assuming that $L=L(\alpha)$ and $b=b(\alpha)$ are continuous monotone
decreasing functions of $\alpha \in (0,1]$ with $\lim_{\alpha \to 0}L =
\lim_{\alpha \to 0}b = \infty$, and that $a=a(\alpha)$ is continuous monotone
increasing with $\lim_{\alpha \to 0}a = -\infty$. The channel width is denoted
by $H(\alpha) = b(\alpha) - a(\alpha)$.

The quantities of interest are the spherically averaged third-order structure
functions associated with the extended fields:
\begin{align}
    \overline{\mathfrak{D}}_{\mathbb{R}}(l) &:= \frac{1}{|\mathcal{D}_\alpha|}
    \mathbb{E}\fint_{\mathbb{S}} \int_{\mathbb{T}_{L(\alpha)} \times \mathbb{R}}
    |\delta_{ln}\omega(x)|^2 (\delta_{ln}u(x) \cdot n) \, dx \, dn,
    \label{def:D_vorticity_ext}\\
    \overline{D}_{\mathbb{R}}(l) &:= \frac{1}{|\mathcal{D}_\alpha|}
    \mathbb{E}\fint_{\mathbb{S}} \int_{\mathbb{T}_{L(\alpha)} \times \mathbb{R}}
    |\delta_{ln}u(x)|^2 (\delta_{ln}u(x) \cdot n) \, dx \, dn.
    \label{def:D_velocity_ext}
\end{align}

\subsubsection{Weak anomalous dissipation  (WAD)}
Applying Itô's formula to the energy and enstrophy functionals yields the stationary balances \cite{YuriAmiraliGigliola},
\begin{align}
    \alpha \mathbb{E}\|u\|_{L^2(\mathcal{D})}^2 + \nu \mathbb{E}\|\nabla u\|_{L^2(\mathcal{D})}^2 &= \varepsilon, \label{eq:energy_balance} \\
    \alpha \mathbb{E}\|\omega\|_{L^2(\mathcal{D})}^2 + \nu \mathbb{E}\|\nabla \omega\|_{L^2(\mathcal{D})}^2 &= \eta. \label{eq:enstrophy_balance}
\end{align}
The WAD framework hypothesizes that in the inviscid limit, all injected energy is dissipated by the large-scale friction $\alpha$, and all injected enstrophy is dissipated by the small-scale viscosity $\nu$. Formally, any family of stationary solutions $\{u_{\nu,\alpha}\}_{\nu,\alpha > 0}$ is assumed to satisfy
\begin{align}
    \lim_{\nu \to 0} \sup_{\alpha \in (0,1)} \nu \mathbb{E}\| \omega\|_{L^2(\mathcal{D})}^2 &= 0, \label{eq:WAD_energy} \\
    \lim_{\alpha \to 0} \sup_{\nu \in (0,1)} \alpha \mathbb{E}\| \omega\|_{L^2(\mathcal{D})}^2 &= 0. \label{eq:WAD_enstrophy}
\end{align}
Let us emphasize that we use the WAD assumption only in this Section, and result of Section \ref{sec:khm_anti} does not use this assumption.

\subsubsection{Dual Cascade}

Assuming WAD, the dual cascade emerges independently of the planet's rotation:

\begin{theorem}[Direct Enstrophy Cascade]\label{thm:enstrophy_cascade} 
Let $\{u_{\nu,\alpha}\}_{\nu,\alpha > 0}$ be a sequence of statistically
stationary solutions satisfying the WAD conditions
\eqref{eq:WAD_energy}-\eqref{eq:WAD_enstrophy} and assume that the forcing satisfies \eqref{eq:noise_high_freq}-\eqref{eq:noise_low_freq}. Then, there exists a sequence
of dissipation scales $l_\nu \to 0$ satisfying
\begin{equation}\label{eq:dissipation_scale_WAD}
    \sup_{\alpha \in (0,1)} \big( \nu \mathbb{E}\|\omega\|_{L^2(\mathcal{D}_\alpha)}^2 \big)^{1/2} = o(l_\nu),
\end{equation}
such that the direct enstrophy cascade laws hold:
\begin{align}
    \lim_{l_I \to 0} \limsup_{\nu,\alpha \to 0} \sup_{l \in [l_\nu, l_I]} \left| \frac{\overline{\mathfrak{D}}_{\mathbb{R}}(l)}{l} + 2\eta \right| &= 0, \label{eq:cascade_mixed} \\
    \lim_{l_I \to 0} \limsup_{\nu,\alpha \to 0} \sup_{l \in [l_\nu, l_I]} \left| \frac{\overline{D}_{\mathbb{R}}(l)}{l^3} - \frac{1}{4}\eta \right| &= 0. \label{eq:cascade_velocity}
\end{align}
\end{theorem}

\begin{theorem}[Inverse Energy Cascade]\label{thm:inverse_cascade} 
Under the same assumptions of Theorem \ref{thm:enstrophy_cascade}, there exists
a sequence of friction scales $l_\alpha \to \infty$ satisfying
\begin{equation}\label{eq:friction_scale_WAD}
    l_\alpha^2 = o\left(\left(\sup_{\nu\in(0,1)}\alpha\mathbb{E}\|\omega\|_{L^2(\mathcal{D}_\alpha)}^2\right)^{-1}\right), 
\end{equation}
such that the inverse energy cascade law holds:
\begin{equation}\label{eq:cascade_energy_inverse}
    \lim_{l_I \to \infty} \limsup_{\nu,\alpha \to 0}  \sup_{l \in [l_I, l_\alpha]} \left| \frac{\overline{D}_{\mathbb{R}}(l)}{l} - 2\varepsilon \right| = 0.
\end{equation}
\end{theorem}

Notice that, unlike the periodic case $\mathbb{T}^2$ \cite{bed2D}, our method cannot prove the \textit{longitudinal} law in the presence of the Coriolis term. This is because applying the KHM framework to the longitudinal structure function $S_3^{\parallel}(l)$ requires controlling the longitudinal Coriolis term
\begin{equation}
    \overline{\varTheta}_{\parallel}(l) := \fint_\mathbb{S} (n \otimes n) : \varTheta(ln) \, dn.
\end{equation}
In contrast to the trace \eqref{eq:Theta_trace_main}, $\overline{\varTheta}_{\parallel}(l)$ does not vanish. A Taylor expansion analysis shows that its dominant error term requires bounding $\mathbb{E}[\|\Delta u\|_{L^2}\|\nabla u\|_{L^2}]$. However, under the WAD framework, we only have the $H^2$-bound $\nu \mathbb{E}\|\Delta u\|_{L^2}^2 \le \eta$. This forces the leading error coefficient to scale as $\nu^{-1/2}$. In the inviscid limit ($\nu \to 0$), this error over the inertial range diverges as $\mathcal{O}(l_I \nu^{-1/2})$. This demonstrates why  deriving the longitudinal law in the presence of the Coriolis force requires more sophisticated machinery.

\subsection{Antisymmetric KHM Balance}\label{sec:khm_anti}

To capture the large-scale formation of zonal anisotropy, we introduce a weighted, \textit{antisymmetric} test function. Let $\mathbb{S}$ be the unit circle and $n = (n^1, n^2) \in \mathbb{S}$ be the direction of the separation vector $y = l n$. We define 
\begin{equation}\label{eq:test_function}
    \phi(y) = \Psi(|y|)\frac{y^2}{|y|} \begin{pmatrix} 0 & 1 \\ -1 & 0 \end{pmatrix},
\end{equation}
where $\Psi$ is a smooth, compactly supported scalar function on $(0, \infty)$. 

Testing the weak KHM equation \eqref{eq:Master_KHM} against \eqref{eq:test_function} annihilates the symmetric quantities of the flow. Specifically, due to the symmetry of the third-order structure tensor $D^j(y)$ and the behavior of the spatial operators under this antisymmetric projection, the nonlinear inertial term, along with the viscous and Ekman dissipations, vanish. Furthermore, the boundary conditions on the divergence-free basis ensure that the stochastic forcing injects strictly zero antisymmetric correlation. The surviving terms yield a novel, exact identity governing the large-scale anisotropic dynamics of the flow.

\begin{theorem}[Geostrophic KHM Balance]\label{thm:antisymmetric_balance}
Let $u$ be a statistically stationary solution to the system \eqref{eq:NS_intro} in the sense of Theorem \ref{thm:stationary_measure}. Assume moreover that the forcing satisfies \eqref{eq:noise_high_freq}. Then the following balance holds for any separation scale $l > 0$:
\begin{align}\label{eq:CHS_antisymmetric}
    0 &= \beta l \mathbb{E} \fint_{\mathbb{S}}\fint_{\mathcal{D}} (n^2)^2 \big( u(x) \cdot u(x+ln) \big) \, dxdn \notag \\
    &\quad - \mathbb{E} \fint_{\mathbb{S}}\fint_{\mathcal{D}} n^2 \big( \nabla p(x) \cdot u^\perp(x+ln) - u^\perp(x) \cdot \nabla p(x+ln) \big) \, dxdn,
 \end{align}
 where $n=(n^1,n^2) \in \mathbb{S}$ is a unit vector.
\end{theorem}
As mentioned in the Introduction, equation \eqref{eq:CHS_antisymmetric} provides the formal statistical counterpart of the large-scale geostrophic balance. 
At small separation scales, the weight $\beta l$ renders the rotational term negligible. Thus, the statistical balance is unaffected by the planetary vorticity gradient $\beta$, recovering local isotropy within the small-scale turbulent cascade.

Conversely, as the separation scale $l$ increases and approaches the \textit{Rhines scale} $l_R \sim (\varepsilon/\beta^3)^{1/5}$ (see Appendix \ref{app:heuristics}), the $\beta$-effect becomes the dominant driving mechanism. In particular, the unbounded growth of $\beta l$ must be compensated by the pressure gradient. This forces a reorganization of the flow. To align the flow into large-scale zonal jets, the fluid suppresses meridional velocity fluctuations (as highlighted in \eqref{eq:supp}). Coupled with the cancellation of the Coriolis terms (Theorem \ref{thm:vanishing_coriolis}), this demonstrates that while planetary rotation preserves the magnitude of the inverse energy cascade, the $\beta$-effect alters the spatial distribution of the energy transfer.

To the best of our knowledge, Theorem \ref{thm:antisymmetric_balance} represents the first exact statistical balance for $\beta$-plane turbulence that captures the emergence of large-scale zonal jets. The proof is given in Section \ref{sec:zonal_dynamics}.

\section{Proofs of the Main Results}\label{sec:proofs}

This section is devoted to the derivations of the theorems stated in Section \ref{sec:main_results}. We proceed as follows: first, we establish the vanishing of the spherically averaged Coriolis terms. Next, we formalize the dual cascade. Finally, we extract the antisymmetric balance.

\subsection{Proof of Theorem \ref{thm:vanishing_coriolis}}\label{sec:vanishing}
We begin by demonstrating that both the enstrophy Coriolis correlation $\mathfrak{Q}(y)$ defined in \eqref{def:coriolis_vort_corr} and the trace of the velocity Coriolis correlation $\text{Tr}(\varTheta(y))$ defined in \eqref{def:Theta} vanish for any statistically stationary incompressible flow in the periodic channel. These results confirm that the Coriolis force does not contribute to the net isotropic transfer of energy or enstrophy at any scale.

Let $u$ be a statistically stationary solution to \eqref{eq:abstract_NS}. As described in Section \ref{sec:weak_khm}, the fields are extended by zero outside the domain $\mathcal{D} = \mathbb{T} \times I$, where $I = [a, b]$. However, since $\operatorname{supp}(u) \subset \bar{\mathcal{D}}$, we restrict our integrations to the channel $\mathcal{D}$, outside of which the integrands identically vanish.

By the $H^2$ regularity established in Theorem \ref{thm:H2_regularity}, there exists a stream function $\psi \in H^3(\mathcal{D})$ such that
\begin{equation}\label{eq:stream_function}
    u = (-\partial_2 \psi, \partial_1 \psi) \quad \text{and} \quad \omega = \Delta \psi.
\end{equation}
Notice that we extend $\psi$ to $\mathbb T\times\mathbb R$ by constants on the exterior of $\mathcal{D}$, so that $\nabla^\perp\psi$ coincides with the zero extension of $u$.

\begin{proof}[Proof of Theorem \ref{thm:vanishing_coriolis}]
We split the proof into two steps. 
Before proceeding, we note that the stream function $\psi$ is periodic in the zonal direction $x^1 \in \mathbb{T}$. Indeed, integrating the divergence-free condition over $\mathbb{T}$ yields
\begin{equation}
    \partial_2 \int_{\mathbb{T}} u^2 dx^1 = 0,
\end{equation}
that is,
\begin{equation}\label{eq:integral}
    \int_{\mathbb{T}} u^2 dx^1 = C
\end{equation}
for $x^2\in[a,b]$. Since $u^2 = 0$ at the boundary $x^2=a$, the integral \eqref{eq:integral} vanishes everywhere, ensuring
\begin{equation}
    \psi(L, x^2) - \psi(0, x^2) = \int_{\mathbb{T}} u^2 dx^1 = 0.
\end{equation}

\noindent\textbf{Step 1: Vanishing of the Coriolis enstrophy flux $\mathfrak{Q}(y)=0$.}\\
Restricting the domain of integration to the periodic channel, the enstrophy Coriolis correlation \eqref{def:coriolis_vort_corr} reads 
\begin{equation}\label{eq:Q_restricted}
    \mathfrak{Q}(y) = \frac{1}{2}\beta \, \mathbb{E} \fint_{\mathbb{T} \times [a, b]} \big( u^2(x) \omega(x+y) + \omega(x) u^2(x+y) \big) \, dx.
\end{equation}
Substituting $\omega = \partial_{11}^2 \psi + \partial_{22}^2 \psi$, the integral decomposes into a zonal component $Z(y)$ and a meridional component $M(y)$:
\begin{align*}
    \fint_{\mathbb{T} \times [a, b]} \big[ \partial_1 \psi(x) \Delta \psi(x+y) + \Delta \psi&(x) \partial_1 \psi(x+y) \big] \, dx\\
    =&\fint_{\mathbb{T} \times [a, b]} \big[ \partial_1 \psi(x) \partial_{11}^2 \psi(x+y) + \partial_{11}^2 \psi(x) \partial_1 \psi(x+y) \big] \, dx\\
    +&\fint_{\mathbb{T} \times [a, b]} \big[ \partial_1 \psi(x) \partial_{22}^2 \psi(x+y) + \partial_{22}^2 \psi(x) \partial_1 \psi(x+y) \big] \, dx\\
    =: &Z(y) + M(y).
\end{align*}
For $Z(y)$, integrating the second term by parts with respect to $x^1$ shifts one derivative from $x$ to $x+y$. Because $\psi$ is periodic in $x^1$, the boundary terms vanish, yielding $Z(y) = 0$.

For $M(y)$, we integrate both terms by parts with respect to $x^2 \in [a, b]$,
\begin{align*}
    \fint_a^b \partial_1 \psi(x) \partial_{22}^2 \psi(x+y) \, dx^2 &= \Big[ \partial_1 \psi(x) \partial_2 \psi(x+y)\Big]_{a}^{b} - \fint_a^b \partial_{12}^2 \psi(x) \partial_2 \psi(x+y) \, dx^2, \\
    \fint_a^b \partial_{22}^2 \psi(x) \partial_1 \psi(x+y) \, dx^2 &= \Big[ \partial_2 \psi(x) \partial_1 \psi(x+y) \Big]_{a}^{b} - \fint_a^b \partial_{2} \psi(x) \partial_{21}^2 \psi(x+y) \, dx^2.
\end{align*}
The no-slip boundary condition \eqref{eq:boundary_conditions} ensures $u^2 = \partial_1 \psi = 0$ and $u^1 = -\partial_2 \psi = 0$ for $x^2 \in \{a, b\}$, canceling the two brackets. We are left with
\begin{equation*}
    M(y) = \fint_{\mathbb{T} \times I} \big[ - \partial_{12}^2 \psi(x) \partial_2 \psi(x+y) - \partial_2 \psi(x) \partial_{21}^2 \psi(x+y) \big] \, dx.
\end{equation*}
Integrating the first term by parts along the periodic zonal direction $x^1$, we get 
\begin{equation*}
    - \fint_a^b dx^2 \fint_{\mathbb{T}} \partial_1 \big(\partial_2 \psi(x)\big) \partial_2 \psi(x+y) \, dx^1 = \fint_{\mathbb{T} \times [a, b]} \partial_2 \psi(x) \partial_{12}^2 \psi(x+y) \, dx,
\end{equation*}
where the boundary term vanishes by periodicity. Since the $H^3$ regularity of the stream function $\psi$ guarantees the commutation of the derivatives in the weak sense $\partial_{12}^2 \psi = \partial_{21}^2 \psi$ in $L^2$, this remaining term cancels the second term of $M(y)$, yielding $M(y) = 0$. Therefore $\mathfrak{Q}(y)=0$ for all $y\in \mathbb{R}^2$, and spherical averaging yields $\overline{\mathfrak{Q}}(l)=0$. 

\text{}

\noindent \textbf{Step 2: Vanishing of the velocity Coriolis trace $\text{Tr}(\varTheta(y))=0$.}\\
Recalling $u^\perp = (-u^2, u^1)$ and the definition \eqref{def:Theta}, we write
\begin{equation}
    \text{Tr}(\varTheta(y)) = \frac{1}{2}\mathbb{E} \fint_{\mathbb{T}\times [a,b]} \big(u(x) \cdot (fu^\perp)(x+y) + (f u^\perp)(x) \cdot u(x+y)\big)\, dx.
\end{equation}
Evaluating the shift $f(x+y) - f(x) = \beta y^2$, the uniform rotation component $f_0$ cancels, and we obtain
\begin{align}
    \text{Tr}(\varTheta(y)) &= \frac{1}{2}\beta y^2 \, \mathbb{E} \fint_{\mathbb{T} \times I} \big( u^2(x)u^1(x+y) - u^1(x)u^2(x+y) \big) \, dx\notag \\
    &= \frac{1}{2}\beta y^2 \, \mathbb{E} \fint_{\mathbb{T} \times I} \big( -\partial_1 \psi(x) \partial_2 \psi(x+y) + \partial_2 \psi(x) \partial_1 \psi(x+y) \big) \, dx.\label{eq:trace_psi}
\end{align}
The no-slip boundary condition \eqref{eq:boundary_conditions} implies $\psi$ is constant for $x^2\in \{a,b\}$. We set $\psi(x^1, a) = 0$. On $x^2=b$, the constant is given by the net zonal transport,
\begin{equation}
    \psi(x^1, b) = - \fint_a^b u^1 dx^2 := C_b.
\end{equation}
Integrating the first term of \eqref{eq:trace_psi} by parts with respect to $x^1$, we derive
\begin{equation}\label{eq:int1}
    \fint_a^b dx^2 \fint_{\mathbb{T}} -\partial_1 \psi (x)\partial_2 \psi(x+y) \, dx^1 = \fint_{\mathbb{T} \times [a, b]} \psi(x) \, \partial_{12}^2 \psi(x+y) \, dx,
\end{equation}
where again by periodicity the boundary term vanishes. Integrating the second term of \eqref{eq:trace_psi} by parts with respect to $x^2$, we get
\begin{align}
    \fint_{\mathbb{T}} dx^1 \fint_a^b \partial_2 \psi(x) \partial_1\psi&(x+y) \, dx^2 \label{eq:int2}\\
    &= \fint_{\mathbb{T}} dx^1 \left( - \fint_a^b \psi(x) \, \partial_{21}^2 \psi(x+y) \, dx^2 + \Big[ \psi(x) \, \partial_1 \psi(x+y) \Big]_{x^2=a}^{x^2=b} \right)\notag \\
    &= - \fint_{\mathbb{T} \times [a, b]} \psi(x) \, \partial_{21}^2 \psi(x+y) \, dx + C_b \fint_{\mathbb{T}} \partial_1 \psi(x^1+y^1, b+y^2) \, dx^1.\notag
\end{align}
Notice that the boundary term evaluated at $x^2=b$ is the integral of a derivative over the periodic domain $\mathbb{T}$. Hence,
\begin{equation*}
    C_b \fint_{\mathbb{T}} \partial_1 \psi(x^1+y^1, b+y^2) \, dx^1 = C_b \big[ \psi(L+y^1, b+y^2) - \psi(y^1, b+y^2) \big] = 0.
\end{equation*}
Summing the two integrated components \eqref{eq:int1} and \eqref{eq:int2}, all boundary terms vanish independently of the constant $C_b$, yielding 
\begin{equation}\label{eq:coclusionvanish}
    \fint_{\mathbb{T} \times [a, b]} \psi(x) \big[ \partial_{12}^2 - \partial_{21}^2 \big] \psi(x+y) \, dx = 0.
\end{equation}
Therefore, $\text{Tr}(\varTheta(y)) = 0$ for all $y \in \mathbb{R}^2$, implying $\overline{\varTheta}(l) = 0$.
\end{proof}

\subsection{Proofs of Theorem \ref{thm:enstrophy_cascade} and Theorem \ref{thm:inverse_cascade}}\label{sec:thermodynamic}
Throughout this section we work in the large-domain limit introduced in Section \ref{sec:main_cascade}.

We require uniform regularity of the stochastic forcing. Assuming that the average injection rates $\varepsilon$ and $\eta$ are finite and independent of $\nu$ and $\alpha$, we require the profiles $\{b_je_j^\alpha\}$ of the noise \eqref{eq:stochasticprocess} to satisfy the following uniform bounds with respect to the expanding domain $\mathcal{D}_\alpha$,
\begin{align}
    \sup_{\alpha \in (0,1]} \sum_{j \in \mathbb{N}}b_j^2 \|\nabla^3 e_j^\alpha\|_{L^2(\mathcal{D}_\alpha)}^2 &\le C < \infty, \label{eq:noise_high_freq} \\
    \lim_{\delta \to 0} \sup_{\alpha \in (0,1]} \sum_{j \in \mathbb{N}} b_j^2 \|(e_j^\alpha)_{\le \delta}\|_{L^2(\mathcal{D}_\alpha)}^2 &= 0.\label{eq:noise_low_freq}
\end{align}
Here, $f_{\le \delta}$ denotes the restriction to frequencies $|k| \le \delta$. As in \cite{bed2D}, condition \eqref{eq:noise_high_freq} guarantees that the forcing $\varphi$ is uniformly bounded in $H^3(\mathcal{D}_\alpha)$, which is necessary to extract the enstrophy rate $\eta$ without diverging remainders. On the other hand, \eqref{eq:noise_low_freq} allows the noise correlation terms to decay to zero at large scales.

To derive the first spherically averaged KHM relation, we evaluate \eqref{eq:Master_KHM} against the test function $\phi(y) = \Phi(|y|)I$. Here, $\Phi$ is a smooth, compactly supported scalar function on $(0, \infty)$, and $I$ denotes the identity matrix. Then, recalling \eqref{def:D_velocity_ext} we have
\begin{align}
    \frac{\overline{D}_{\mathbb{R}}(l)}{l} &= -\frac{4\nu}{l}\overline{\Gamma}'(l) + \frac{4\alpha}{l^2}\int_0^l r\overline{\Gamma}(r) \, dr - \frac{4}{l^2}\int_0^l r\overline{a}(r) \, dr, \label{eq:sf_velocity_master}
\end{align}
where from definitions \eqref{def:Gamma} and \eqref{def:a}, we get 
\begin{align}
    \overline{\Gamma}(l)&=\fint_{\mathbb{S}}\trace\Gamma(ln)dn,\\
    \overline{a}(l)&=\fint_{\mathbb{S}}\trace a(ln)dn.
\end{align}
Similarly, starting from the vorticity balance \eqref{eq:vorticity_khm} and integrating both sides over $\{|y|\leq l\}$, we derive 
\begin{align}
    \frac{\overline{\mathfrak{D}}_{\mathbb{R}}(l)}{l} &= -\frac{4\nu}{l}\overline{\mathfrak{C}}'(l) + \frac{4\alpha}{l^2}\int_0^l r\overline{\mathfrak{C}}(r) \, dr - \frac{4}{l^2}\int_0^l r\overline{\mathfrak{a}}(r) \, dr, \label{eq:sf_vv}
\end{align}
where $\overline{\mathfrak{D}}_{\mathbb{R}}$ is defined in \eqref{def:D_vorticity_ext}
and from \eqref{def:vorticity_corr} and \eqref{def:noise_curl_corr} we have
\begin{align}
    \overline{\mathfrak{C}}(l)&=\fint_{\mathbb{S}}\mathfrak{C}(ln)dn,\\
    \overline{\mathfrak{a}}(l)&=\fint_{\mathbb{S}}\mathfrak{a}(ln)dn.
\end{align}
As we said in Section \ref{rem:quadratic_vs_cubic}, the correlations $\overline{\Gamma}$, $\overline{a}$, $\overline{\mathfrak{C}}$, $\overline{\mathfrak{a}}$ appearing on the right-hand sides of \eqref{eq:sf_velocity_master} and \eqref{eq:sf_vv}
may equivalently be evaluated on $\mathcal{D}_\alpha$. On
the other hand, the left-hand sides are the extended fluxes
\eqref{def:D_vorticity_ext}-\eqref{def:D_velocity_ext}. For full details on
the derivation of \eqref{eq:sf_velocity_master} and \eqref{eq:sf_vv}, see
\cite{YuriAmiraliGigliola}.

\begin{proof}[Proof of Theorems \ref{thm:enstrophy_cascade} and \ref{thm:inverse_cascade}]
Due to Theorem \ref{thm:vanishing_coriolis}, the spherically averaged Coriolis terms vanish identically, so that the balances \eqref{eq:sf_velocity_master}
and \eqref{eq:sf_vv} reduce to the non-rotating relations considered in
\cite{bed2D}. Since these balances are formulated on the translation-invariant
domain $\mathbb{T}_{L(\alpha)} \times \mathbb{R}$, taking $\nu, \alpha \to 0$ yields the dual cascade laws exactly as in \cite{bed2D}, using the energy and enstrophy balance \eqref{eq:energy_balance}, \eqref{eq:enstrophy_balance}. Finally, dividing the velocity balance \eqref{eq:sf_velocity_master} by $l^2$ recovers the scaling
\eqref{eq:cascade_velocity} of the direct enstrophy cascade.
\end{proof}

\subsubsection{Boundary layers and the fluxes on the channel} \label{sec: boundary layer}
In Section \ref{rem:quadratic_vs_cubic}, we explained that we have two choices to average quantities such as $D^j,\Gamma$, etc.. We also highlighted that the choice of averaging is only consequential for  $D^j$ and $\mathfrak{D}$. In fact, one can define:  
\begin{align}
    \overline{\mathfrak{D}}_{\alpha}(l) &:= \mathbb{E}\fint_{\mathbb{S}}
    \fint_{\mathcal{D}_\alpha} |\delta_{ln}\omega(x)|^2 (\delta_{ln}u(x) \cdot n)
    \, dx \, dn, \label{def:D_vorticity_phys}\\
    \overline{D}_{\alpha}(l) &:= \mathbb{E}\fint_{\mathbb{S}} \fint_{\mathcal{D}_\alpha}
    |\delta_{ln}u(x)|^2 (\delta_{ln}u(x) \cdot n) \, dx \, dn.
    \label{def:D_velocity_phys}
\end{align}
Notice that both pairs $\overline{\mathfrak{D}}_{\alpha}, \overline{\mathfrak{D}}_{\mathbb{R}}$, and $\overline{D}_{\alpha}, \overline{D}_{\mathbb{R}}$ are normalized by $|\mathcal{D}_\alpha|$, so that they are directly
comparable. Let us emphasize that the discrepancy among above mentioned pairs stems from the fact that we extend the fluid by zero outside the channel; therefore, the mentioned difference is solely arises from the behavior of the fluid inside a narrow boundary layer. Such a boundary layer has a thickness of at most $l_{\alpha}$ (for large scale quantities) and $l_{I}$ for small scale quantities. Taking $H(\alpha) \to \infty$ faster than $l_{\alpha}$ (which is the correct choice) means this boundary layer is negligible macroscopically. 

Moreover, from a physical perspective, we are interested in the cascade phenomena in the bulk and not the boundary effects. Effect of the boundary on the Navier Stokes equation has an extensive literature (see \cite{BardosTitiWiedemann2019,DrivasNguyen2018, Kato1984}) and it is out of the scope of our work. 

More importantly, we have chosen this boundary to reflect the fact that in the flow the symmetry is broken. However, in a more physical picture of a fluid on a rotating planet such boundary is fictitious. Consequently, we believe that such discrepancy should be considered as an artifact of the domain and not a physical quantity. 
Finally, let us mention that in our previous work \cite{YuriAmiraliGigliola} the assumption that the  boundary layer's contribution is negligible  was implicit.

Still, we provide the following quantitative comparison. Define
\begin{equation}\label{def:K}
    \mathcal{K}(l) := \big|\overline{D}_{\mathbb{R}}(l) - \overline{D}_{\alpha}(l)\big|,
    \qquad
    \mathcal{K}_{\omega}(l) := \big|\overline{\mathfrak{D}}_{\mathbb{R}}(l) - \overline{\mathfrak{D}}_{\alpha}(l)\big|.
\end{equation}
Since $u$ and $\omega$ vanish outside $\mathcal{D}_\alpha$, for $x \notin
\mathcal{D}_\alpha$ the increments reduce to $\delta_{ln}u(x) = u(x+ln)$ and
$\delta_{ln}\omega(x) = \omega(x+ln)$, so that the excess support is confined
to the exterior boundary layers
\begin{equation}
    \widetilde{\mathcal{B}}_{l,n}=\left\{x\notin\mathcal{D}_\alpha;\, x+ln \in \mathcal{D}_\alpha\right\}
\end{equation}
adjacent to $x^2 = a(\alpha)$ and $x^2 = b(\alpha)$. Hence
\begin{equation}
    \mathcal{K}(l) \le \frac{1}{|\mathcal{D}_\alpha|}\mathbb{E} \fint_\mathbb{S}\int_{\mathcal{B}_{l,n}} |u(x+ln)|^3 \, dx\,dn = \frac{1}{|\mathcal{D}_\alpha|}\mathbb{E} \fint_\mathbb{S}\int_{\widetilde{\mathcal{B}}_{l,n}} |u(y)|^3 \, dy\,dn,
\end{equation}
where $y=x+ln$ and $\widetilde{\mathcal{B}}_{l,n} \subset \mathcal{D}_\alpha$
is the shifted interior boundary layer. Setting
\begin{equation}
    \widetilde{\mathcal{B}}_l=\left\{y\in \mathcal{D}_\alpha;\, \text{dist}(y,\partial \mathcal{D}_\alpha)\leq l\right\},
\end{equation}
we have $\widetilde{\mathcal{B}}_{l,n} \subset \widetilde{\mathcal{B}}_{l}$ and
$|\widetilde{\mathcal{B}}_{l}| = 2lL(\alpha)$, so that
\begin{equation}\label{eq:boundary_error_exact}
   \mathcal{K}(l)\le  \frac{1}{|\mathcal{D}_\alpha|}\mathbb{E} \int_{\widetilde{\mathcal{B}}_{l}} |u(y)|^3 \, dy =  \frac{|\widetilde{\mathcal{B}}_{l}|}{|\mathcal{D}_\alpha|}  \mathbb{E}\|u\|_{L^3(\widetilde{\mathcal{B}}_{l})}^3,
\end{equation}
where, consistently with \eqref{eq:norms}, the norm
$\|\cdot\|_{L^3(\widetilde{\mathcal{B}}_{l})}$ is normalized by
$|\widetilde{\mathcal{B}}_{l}|$. Since
$|\widetilde{\mathcal{B}}_{l}|/|\mathcal{D}_\alpha| = 2l/H(\alpha)$, we obtain
\begin{equation}\label{eq:error}
    \mathcal{K}(l) \leq   \frac{2l}{H(\alpha)} \mathbb{E}\|u\|_{L^3(\widetilde{\mathcal{B}}_{l})}^3,
\end{equation}
and, by the same argument applied to \eqref{def:D_vorticity_ext},
\begin{equation}\label{eq:error_vortc}
    \mathcal{K}_{\omega}(l) \leq   \frac{2l}{H(\alpha)} \mathbb{E}\left(\|\omega\|_{L^3(\widetilde{\mathcal{B}}_l)}^2\|u\|_{L^3(\widetilde{\mathcal{B}}_l)}\right).
\end{equation}

\begin{cor}\label{cor:channel}
Let $\{u_{\nu,\alpha}\}_{\nu,\alpha>0}$ satisfy the hypotheses of
Theorem~\ref{thm:enstrophy_cascade}, and let $l_\nu$ and $l_\alpha$
be scales satisfying \eqref{eq:dissipation_scale_WAD} and
\eqref{eq:friction_scale_WAD}, respectively. Then:
\begin{enumerate}
    \item[(i)] if
    $\displaystyle \lim_{l_I \to 0}\limsup_{\nu,\alpha \to 0} \sup_{l \in [l_\nu, l_I]} \frac{\mathcal{K}_{\omega}(l)}{l} = 0$,
    then \eqref{eq:cascade_mixed} holds with $\overline{\mathfrak{D}}_{\mathbb{R}}$ replaced by $\overline{\mathfrak{D}}_{\alpha}$;
    \item[(ii)] if
    $\displaystyle \lim_{l_I \to 0}\limsup_{\nu,\alpha \to 0} \sup_{l \in [l_\nu, l_I]} \frac{\mathcal{K}(l)}{l^3} = 0$,
    then \eqref{eq:cascade_velocity} holds with $\overline{D}_{\mathbb{R}}$ replaced by $\overline{D}_{\alpha}$;
    \item[(iii)] if
    $\displaystyle \lim_{l_I \to \infty}\limsup_{\nu,\alpha \to 0} \sup_{l \in [l_I, l_\alpha]} \frac{\mathcal{K}(l)}{l} = 0$,
    then \eqref{eq:cascade_energy_inverse} holds with $\overline{D}_{\mathbb{R}}$ replaced by $\overline{D}_{\alpha}$.
\end{enumerate}
\end{cor}

\begin{proof}
Each statement follows from the triangle inequality. For instance, in case
(iii),
\begin{equation*}
    \left|\frac{\overline{D}_{\alpha}(l)}{l} - 2\varepsilon\right| \le \left|\frac{\overline{D}_{\mathbb{R}}(l)}{l} - 2\varepsilon\right| + \frac{\mathcal{K}(l)}{l},
\end{equation*}
and one takes the supremum over $l \in [l_I, l_\alpha]$ followed by the limits.
Cases (i) and (ii) are identical, with the normalizations $l$ and $l^3$, respectively.
\end{proof}

\begin{remark}\label{rem:boundary_admissibility}
The hypotheses of Corollary \ref{cor:channel} are quantitative
non-concentration assumptions near the rigid walls. By \eqref{eq:error} and
\eqref{eq:error_vortc}, they are implied by
\begin{align}
    \sup_{l \in [l_\nu, l_I]}\mathbb{E}\left(\|\omega\|_{L^3(\widetilde{\mathcal{B}}_l)}^2 \|u\|_{L^3(\widetilde{\mathcal{B}}_l)}\right) &= o(H(\alpha)),\label{eq:admissibility_condition_small_vorticity}\\
    \sup_{l \in [l_\nu, l_I]} l^{-2}\,\mathbb{E}\|u\|_{L^3(\widetilde{\mathcal{B}}_l)}^3 &= o(H(\alpha)),\label{eq:admissibility_condition_small_velocity}\\
    \sup_{l \in [l_I, l_\alpha]}\mathbb{E}\|u\|_{L^3(\widetilde{\mathcal{B}}_l)}^3 &= o(H(\alpha)),\label{eq:admissibility_condition_large_velocity}
\end{align}
as $\alpha \to 0$, respectively for (i), (ii) and (iii). These conditions are
not purely geometric: they require the meridional width $H(\alpha)$ to grow
fast enough relative to the possible concentration of third-order velocity and
vorticity moments in the strips $\widetilde{\mathcal{B}}_l$, whose relative
thickness is $2l/H(\alpha)$. The additional factor $l^{-2}$ in
\eqref{eq:admissibility_condition_small_velocity} reflects the $l^3$ scaling of
the direct velocity law, and makes it the most demanding of the three.
A slightly stronger but more symmetric sufficient condition follows from Young's
inequality $|\omega|^2|u| \le \tfrac{2}{3}|\omega|^3 + \tfrac{1}{3}|u|^3$,
namely
\begin{equation}\label{eq:admissibility_condition_symmetric}
    \sup_{l \in [l_\nu, l_\alpha]}\mathbb{E}\left( \|u\|_{L^3(\widetilde{\mathcal{B}}_l)}^3 + \|\omega\|_{L^3(\widetilde{\mathcal{B}}_l)}^3 \right) = o(H(\alpha)).
\end{equation}
Finally, we note that in the doubly periodic setting of \cite{bed2D} one has
$\mathcal{K} \equiv \mathcal{K}_\omega \equiv 0$, so that Corollary
\ref{cor:channel} is trivial there and Theorems \ref{thm:enstrophy_cascade} and
\ref{thm:inverse_cascade} reduce exactly to the corresponding statements of
\cite{bed2D}.
\end{remark}

\subsection{Proof of Theorem \ref{thm:antisymmetric_balance}}\label{sec:zonal_dynamics}

In this section, we provide the proof of Theorem \ref{thm:antisymmetric_balance}. To capture the large-scale zonal dynamics of the $\beta$-plane, we depart from the classical isotropic framework. Instead, we derive a novel antisymmetric Kármán-Howarth-Monin balance by using the test function \eqref{eq:test_function}.

\begin{proof}[Proof of Theorem \ref{thm:antisymmetric_balance}]
Let $\mathbb{S}$ be the unit circle and let $y = l n$ be the separation vector, with $l = |y|$ and $n = (n^1, n^2) \in \mathbb{S}$. We introduce the matrix 
\begin{equation*}
    J = \left(\begin{matrix} 0 & 1 \\ -1 & 0 \end{matrix}\right).
\end{equation*}
The Frobenius product of any $A$ with $J$ is given by
\begin{equation}\label{eq:frobenius}
    A : J = \text{Tr}(A^T J) = A_{12} - A_{21}.
\end{equation}
We test the weak KHM identity \eqref{eq:Master_KHM} against the weighted, antisymmetric test function defined in \eqref{eq:test_function}, which can be rewritten as
\begin{equation}
    \phi(y) = \Psi(|y|) \frac{y^2}{|y|} J = \Psi(l) n^2 J,
\end{equation}
where $\Psi(l)$ is a smooth scalar function compactly supported on $(0, \infty)$. Because $J$ is a constant antisymmetric matrix, $\phi(y)$ is an antisymmetric tensor field.

\text{}

\noindent \textbf{Step 1: Non-Linear Term.}\\
The third-order structure tensor $D^j(y)$ defined in \eqref{def:Dj} is symmetric in its components by construction ($\delta_y u \otimes \delta_y u$). Since $\phi(y)$ is antisymmetric, its derivative $\partial_{y^j}\phi(y)$ is also an antisymmetric matrix. The Frobenius product of an antisymmetric matrix and a symmetric matrix is zero. Therefore, the non-linear term vanishes at all scales,
\begin{equation}\label{eq:anti_inertial}
    \sum_{j=1}^2 \partial_{y^j}\phi(y) : D^j(y) = 0.
\end{equation}

\noindent \textbf{Step 2: Viscous and Ekman Dissipation.}\\
We evaluate the antisymmetric projection of the velocity correlation tensor $\Gamma(y)$. Using $u^1 = -\partial_2 \psi$ and $u^2 = \partial_1 \psi$, the projection is
\begin{align*}
    \Gamma_{12}(y) - \Gamma_{21}(y) &= \mathbb{E}\fint_{\mathcal{D}} \big( u^1(x)u^2(x+y) - u^2(x)u^1(x+y) \big)\, dx \\
    &= \mathbb{E}\fint_{\mathcal{D}} \big( -\partial_2\psi(x)\partial_1\psi(x+y) + \partial_1\psi(x)\partial_2\psi(x+y) \big)\, dx.
\end{align*}
We integrate by parts to transfer the derivatives from $x$ to $x+y$. As established in \eqref{eq:trace_psi}-\eqref{eq:coclusionvanish}, the boundary terms vanish, and we obtain 
\begin{equation}
    \Gamma(y) : J = \mathbb{E}\fint_{\mathcal{D}} \psi(x) \big[ \partial_{21}^2 - \partial_{12}^2 \big] \psi(x+y) \, dx = 0.
\end{equation}
Since the Laplacian operator $\Delta$ preserves the antisymmetric matrix structure $J$, both the viscous dissipation $\Delta \phi(y)$ and the Ekman damping $\phi(y)$ are proportional to $J$. Thus, $\Delta \phi : \Gamma = 0$ and $\phi : \Gamma = 0$. Both dissipation mechanisms vanish from the balance.

\text{}

\noindent \textbf{Step 3: Noise Term.}\\
For the noise tensor $a(y)$ defined in \eqref{def:a}, its antisymmetric projection yields
\begin{equation}
    a(y) : J = \frac{1}{2} \sum_j b_j^2 \fint_{\mathcal{D}} \big( e_{j}^{1}(x)e_{j}^{2}(x+y) - e_{j}^2(x)e_{j}^1(x+y) \big) \, dx.
\end{equation}
The orthonormal basis functions $\{e_j\}$ are divergence-free and satisfy the no-slip boundary conditions. By the spatial regularity of the stochastic forcing required in \eqref{eq:noise_high_freq}, we have that $e_j \in H^3(\mathcal{D})$. Thus, there exist scalar stream functions $\psi_j \in H^4(\mathcal{D})$ such that $e_j = (-\partial_2 \psi_j, \partial_1 \psi_j)$. Each $\psi_j$ is extended to $\mathbb T\times\mathbb R$ by its boundary constants on the exterior components, so that $\nabla^\perp\psi_j$ agrees with the zero extension of $e_j$.
Furthermore, the boundary conditions imply that each $\psi_j$ is constant on the boundaries $x^2 \in \{a,b\}$. Therefore, we obtain
\begin{equation}
    a(y) : J = \frac{1}{2} \sum_j b_j^2 \fint_{\mathcal{D}} \big( -\partial_2\psi_j(x)\partial_1\psi_j(x+y) + \partial_1\psi_j(x)\partial_2\psi_j(x+y) \big) \, dx.
\end{equation}
This integral is equivalent to the antisymmetric projection of the velocity correlation evaluated in Step 2. Applying a similar integration by parts, the boundary terms vanish due to the boundary conditions. Thus, the mixed derivatives commute and cancel, yielding
\begin{equation}
    a(y) : J = \frac{1}{2} \sum_j b_j^2 \fint_{\mathcal{D}} \psi_j(x) \big[ \partial_{21}^2 - \partial_{12}^2 \big] \psi_j(x+y) \, dx = 0.
\end{equation}
Therefore, the stochastic forcing injects zero antisymmetric correlation at all separation scales.

\text{}

\noindent \textbf{Step 4: Coriolis and Pressure Terms.}\\
For the Coriolis tensor $\varTheta(y)$ defined in \eqref{def:Theta}, 
evaluating the tensor products and using 
$f(x+y)-f(x)=\beta y^2$, we obtain
\begin{equation}
    \varTheta(y) : J = \frac{1}{2}\beta y^2 \, \mathbb{E} \fint_{\mathcal{D}} \left(
    u^1(x)u^1(x+y) + u^2(x)u^2(x+y)\right)\,dx .
\end{equation}
Equivalently, if $y=ln$ with $n\in\mathbb S$, then
\begin{equation}\label{eq:anti_coriolis}
    \varTheta(ln) : J = \frac{1}{2} \beta l n^2 \, \mathbb{E} \fint_{\mathcal{D}} \left( u^1(x)u^1(x+ln) + u^2(x)u^2(x+ln) \right) dx.
\end{equation}
For the pressure tensor $\Pi(y)$ defined in \eqref{def:Pi}, we evaluate the 
antisymmetric parts of 
$A=\nabla p(x)\otimes u(x+y)$ and 
$B=u(x)\otimes \nabla p(x+y)$. Recalling that 
$u^\perp=(-u^2,u^1)$ and using \eqref{eq:frobenius}, we find
\begin{equation*}
    A_{12}-A_{21}=-\nabla p(x)\cdot u^\perp(x+y),
\end{equation*}
and
\begin{equation*}
    B_{12}-B_{21}=u^\perp(x)\cdot \nabla p(x+y).
\end{equation*}
Therefore,
\begin{equation}
    \Pi(y):J=-\frac12\mathbb{E}\fint_{\mathcal D}\left(\nabla p(x)\cdot u^\perp(x+y)-u^\perp(x)\cdot\nabla p(x+y)\right)dx.
\end{equation}
Equivalently, for $y=ln$,
\begin{equation}\label{eq:anti_pressure}
    \Pi(ln):J=-\frac12 \mathbb{E}\fint_{\mathcal D}\left( \nabla p(x)\cdot u^\perp(x+ln)-u^\perp(x)\cdot\nabla p(x+ln)\right)dx.
\end{equation}

\text{}

\noindent \textbf{Step 5: Final Assembly.}\\
Recalling the coefficients of the weak KHM identity \eqref{eq:Master_KHM}, the remaining terms of the balance reduce to
\begin{equation}
    -4\int_{\mathbb{R}^2} \phi(y):\varTheta(y) \, dy - 4\int_{\mathbb{R}^2} \phi(y):\Pi(y) \, dy  = 0.
\end{equation}
Passing to polar coordinates $dy = l \, dl \, dn$, and factoring out the  test function $\Psi(l)$, we obtain 
\begin{equation}
    \int_0^\infty \Psi(l) l \left[ \fint_{\mathbb{S}} n^2 \big( -4\varTheta(ln):J - 4\Pi(ln):J \big) dn \right] dl = 0.
\end{equation}
Since this balance holds for any smooth, compactly supported test function $\Psi(l)$, the bracket term vanishes pointwise in $l$. Substituting \eqref{eq:anti_coriolis} and \eqref{eq:anti_pressure} into the bracket gives
\begin{align*}
    &-4 \left( \frac{1}{2} \beta l \mathbb{E} \fint_{\mathbb{S}}\fint_{\mathcal{D}} (n^2)^2  (u(x) \cdot u(x+ln)) dx \, dn\right) \\
    &-4 \left( - \frac{1}{2} \, \mathbb{E} \fint_{\mathbb{S}} \fint_{\mathcal{D}} n^2\big( \nabla p(x) \cdot u^\perp(x+ln) - u^\perp(x) \cdot \nabla p(x+ln) \big) dx\,dn\right) = 0.
\end{align*}
Dividing the equation by $-2$ yields the geostrophic identity \eqref{eq:CHS_antisymmetric}, concluding the proof.
\end{proof}

\appendix

\section{Heuristics of Zonal Jets and the Rhines Scale}\label{app:heuristics}

In this Appendix, we connect the exact antisymmetric KHM identity to the classical phenomenological heuristics of two-dimensional geophysical turbulence. For the reader's convenience, let us recall the balance derived in Theorem \ref{thm:antisymmetric_balance}:
\begin{equation}\label{eq:app_balance}
    0 = \beta l \, \mathbb{E}[\mathcal{E}_{zonal}(l)] - \mathbb{E}[\mathcal{P}_{\perp}(l)],
\end{equation}
where
\begin{align}
    \mathcal{E}_{zonal} (l)&= \fint_{\mathbb{S}}\fint_{\mathcal{D}} (n^2)^2 \big( u(x) \cdot u(x+ln) \big) \, dxdn, \label{eq:zonal_energy} \\
    \mathcal{P}_{\perp}(l) &= \fint_{\mathbb{S}}\fint_{\mathcal{D}} n^2 \big( \nabla p(x) \cdot u^\perp(x+ln) - u^\perp(x) \cdot \nabla p(x+ln) \big) \, dxdn.
\end{align}

\paragraph{Small Scales ($l \to 0$): Isotropic Regime.} 
For a small separation $l$, the Coriolis weight $\beta l$ becomes negligible. To preserve the exact identity \eqref{eq:app_balance}, the antisymmetric pressure correlation must vanish. This cancellation physically corresponds to the recovery of local isotropy. The variation of the planetary rotation $\beta$ does not constrain the flow, confirming that the small-scale enstrophy cascade is independent of the Coriolis force, in agreement with our previous result \cite{YuriAmiraliGigliola}.

\paragraph{Large Scales ($l \to \infty$): Emergence of Jets.}
As the inverse cascade transfers energy to larger scales, the parameter $\beta l$ grows, forcing the fluid dynamics to generate a proportional pressure correlation to maintain the statistical balance.

To understand why pressure gradients dominate at large scales, we recall the standard dimensional analysis of the velocity system \eqref{eq:NS_intro} \cite{vallis1}. Let $U$, $L$, and $T$ be the characteristic velocity, spatial, and temporal scales of the flow. Under the mid-latitude assumption $\beta H \ll f_0$, where $H=b-a$ is the channel width, the magnitude of the Coriolis force is dominated by the uniform planetary rotation, scaling as $f(x^2) u^\perp \sim f_0 U$. The non-linear inertial term scales as $(u \cdot \nabla)u \sim U^2/L$, the time derivative as $\partial_t u \sim U/T$, the Ekman drag as $\alpha u \sim \alpha U$, and the viscous dissipation as $\nu \Delta u \sim \nu U/L^2$.  

The ratio of the inertial forces to the Coriolis force is quantified by the non-dimensional \textit{Rossby number}, $Ro = U / (f_0 L)$. 
As energy is transferred to large scales ($L \to \infty$), the Rossby number vanishes ($Ro \to 0$), making the inertial and viscous terms negligible. Furthermore, macroscopic geophysical flows are characterized by slow temporal evolution ($1/(f_0 T) \ll 1$) and weak friction ($\alpha/f_0 \ll 1$). Thus, the inertial, viscous, temporal, and frictional terms all become of lower order. To preserve the momentum balance, the fluid generates a pressure gradient $\nabla p \sim f_0 U$. Formally, taking this limit collapses the Navier-Stokes equations onto the geostrophic balance
\begin{equation}\label{eq:GB_intro}
    (f_0 + \beta x^2) u^\perp \approx -\nabla p.
\end{equation}

\paragraph{The Rhines Scale and the Formation of Zonal Jets.}
While the uniform component $f_0$ dictates the macroscopic geostrophic balance \eqref{eq:GB_intro}, uniform rotation cannot generate internal anisotropy. The antisymmetric KHM identity \eqref{eq:app_balance} translates this physical constraint into a two-point statistical balance. The antisymmetric projection annihilates $f_0$, revealing that the relevant transition is driven exclusively by $\beta$.

To capture this physical transition, we perform a dimensional scaling analysis on the terms of the exact identity \eqref{eq:app_balance}. In the inertial range of the inverse energy cascade, the classical Kolmogorov-Kraichnan laws predict the characteristic velocity fluctuations at scale $l$ to behave as $U_l \sim (\varepsilon l)^{1/3}$ \cite{kolmogorov2, Kraichnan}. Consequently, the macroscopic spatial velocity correlation scales as $\mathbb{E}[\mathcal{E}_{zonal}(l)] \sim U_l^2 \sim (\varepsilon l)^{2/3}$. The Coriolis term in the balance \eqref{eq:app_balance} therefore scales as:
\begin{equation}\label{eq:scaling_coriolis}
    \beta l \cdot \mathbb{E}[\mathcal{E}_{zonal}(l)] \sim \beta l (\varepsilon l)^{2/3} = \beta \varepsilon^{2/3} l^{5/3}.
\end{equation}
Simultaneously, the pressure behaves as a specific energy, scaling quadratically with the velocity $P_l \sim U_l^2 \sim (\varepsilon l)^{2/3}$. The gradient scales as $\nabla p \sim P_l/l$. Therefore, the antisymmetric cross-correlation between the pressure gradient and the orthogonal velocity absorbs the constant energy injection rate $\varepsilon$,
\begin{equation}
    \mathbb{E}[\mathcal{P}_\perp(l)] \sim \frac{P_l}{l} U_l \sim \left(\frac{(\varepsilon l)^{2/3}}{l}\right) (\varepsilon l)^{1/3} = \varepsilon.
\end{equation}
The critical transition defining the onset of anisotropy occurs when the growing $\beta$-term overtakes the inertial pressure-vorticity correlation. Equating the two yields
\begin{equation}\label{eq:balance_beta_eps}
    \beta \varepsilon^{2/3} l^{5/3} \sim \varepsilon \implies l \sim \left( \frac{\varepsilon}{\beta^3} \right)^{1/5} := l_R.
\end{equation}
By substituting $\varepsilon \sim U_l^3 / l$, equation \eqref{eq:balance_beta_eps} simplifies to
\begin{equation}
    l_R \sim \sqrt{\frac{U_l}{\beta}},
\end{equation}
successfully recovering the Rhines scale from the exact antisymmetric KHM balance.

For scales $l \gtrsim l_R$, the geometrical growth of the $\beta$-term dominates the balance. To understand how the system rearranges to compensate for this growth and align the flow into \textit{zonal jets}, we apply the two-dimensional curl ($\nabla^\perp \cdot$) to the macroscopic geostrophic balance \eqref{eq:GB_intro}. This operation annihilates the pressure gradient ($\nabla^\perp \cdot \nabla p = 0$), yielding
\begin{equation}
    \nabla^\perp \cdot (f u^\perp) = \partial_1 (f u^1) + \partial_2 (f u^2) \approx 0.
\end{equation}
By expanding the derivatives and using the incompressibility condition ($\partial_1 u^1 + \partial_2 u^2 = 0$), the contribution of the uniform rotation vanishes ($f_0 \nabla \cdot u = 0$). The only surviving term is the planetary vorticity gradient:
\begin{equation}\label{eq:supp}
    \beta u^2 \approx 0.
\end{equation}
In the absence of a wind stress curl \cite{vallis1}, this unforced geostrophic vorticity balance (representing the \textit{homogeneous Sverdrup relation}) forces the suppression of the meridional velocity fluctuations ($u^2 \approx 0$). By incompressibility, this implies $\partial_1 u^1 \approx 0$, so that the dominant velocity field is expected to be nearly zonal and approximately independent of the zonal coordinate,
\begin{equation}
    u(x) \approx (U(x^2),0).
\end{equation}
Thus the surviving component is the zonal velocity, but its amplitude may vary 
with the meridional coordinate. This is precisely the banded structure of zonal 
jets, and the antisymmetric KHM balance \eqref{eq:app_balance} captures this 
large-scale anisotropic organization.

\section{Non-Triviality of the Antisymmetric Balance}\label{appB} 
To  demonstrate that the antisymmetric KHM identity \eqref{eq:app_balance} is non-trivial, we can evaluate the asymptotic behavior of the Coriolis correlation in the limit of vanishing large-scale friction $\alpha \to 0$. Notice that this regime is the relevant one for isolating large-scale dynamics (see, e.g., Theorem 1.21 of \cite{bed2D}). 

Recall the zonal energy correlation defined in \eqref{eq:zonal_energy},
\begin{equation}
    \beta l \, \mathbb{E}[\mathcal{E}_{zonal}(l)] = \beta l \, \mathbb{E} \fint_{\mathbb{S}}\fint_{\mathcal{D}} (n^2)^2 \big( u(x) \cdot u(x+ln) \big) \, dx \, dn.
\end{equation}
By adding and subtracting $u(x) \cdot u(x)$ inside the inner product, we decompose the correlation as
\begin{equation}
    |u(x)|^2 + u(x) \cdot \delta_{ln}u(x).
\end{equation}
Substituting this decomposition into the integral, and noting that the spatial average of $|u(x)|^2$ is independent of $n$, the integration over the unit circle yields $\fint_{\mathbb{S}} (n^2)^2 dn = 1/2$, we obtain
\begin{equation}\label{eq:app_decomposition}
    \beta l \, \mathbb{E}[\mathcal{E}_{zonal}(l)] = \frac{1}{2}\beta l \, \mathbb{E}\|u\|_{L^2(\mathcal{D})}^2 + \beta l \, \mathbb{E} \fint_{\mathbb{S}}\fint_{\mathcal{D}} (n^2)^2 \big( u(x) \cdot \delta_{ln}u(x) \big) \, dx \, dn.
\end{equation}
We evaluate this expression in the regime relevant for large-scale dynamics. We fix the small-scale viscosity $1\gg\nu > 0$  sufficiently small, our choice of $\nu>0$ becomes clear later. Then, we take the limit $\alpha \to 0$. From the energy balance and enstrophy balance \eqref{eq:energy_balance} \eqref{eq:enstrophy_balance}, and recalling $\|\nabla u \|_{L^2} = \|\omega \|_{L^2}$, we deduce that  the dominant contribution to the kinetic energy is given by $\mathbb{E}\|u\|_{L^2}^2 \sim \varepsilon/\alpha$. Therefore, the first term on the right-hand side of \eqref{eq:app_decomposition} diverges as $\mathcal{O}(\alpha^{-1})$, using the WAD assumption \eqref{eq:WAD_energy} and taking $\nu$ small enough.

For the second term, we apply the Cauchy-Schwarz inequality, along with the bound for the increment $\|\delta_{ln}u\|_{L^2} \le l\|\nabla u\|_{L^2}$. This yields the upper bound:
\begin{align}
    \left| \beta l \, \mathbb{E} \fint_{\mathbb{S}}\fint_{\mathcal{D}} (n^2)^2 \big( u(x) \cdot \delta_{ln}u(x) \big) \, dx \, dn \right| 
    &\le \beta l \, \mathbb{E} \Big[ \|u\|_{L^2(\mathcal{D})} \|\delta_{ln}u\|_{L^2(\mathcal{D})} \Big] \notag \\
    &\le \beta l^2 \big(\mathbb{E}\|u\|_{L^2(\mathcal{D})}^2\big)^{1/2} \big(\mathbb{E}\|\nabla u\|_{L^2(\mathcal{D})}^2\big)^{1/2}.
\end{align}
From the energy balance \eqref{eq:energy_balance}, the gradient is uniformly bounded by $\mathbb{E}\|\nabla u\|_{L^2}^2 \le \varepsilon/\nu$. Combining this with the energy bound, the mixed term is bounded by $\mathcal{O}(\alpha^{-1/2})$.

As $\alpha \to 0$, the $\mathcal{O}(\alpha^{-1})$ growth dominates the $\mathcal{O}(\alpha^{-1/2})$ bound of the increment. Therefore, for any fixed $l > 0$, the zonal energy correlation behaves as
\begin{equation}
    \beta l \, \mathbb{E}[\mathcal{E}_{zonal}(l)] \sim \frac{\varepsilon \beta l}{2\alpha}.
\end{equation}
This proves that the identity \eqref{eq:app_balance} cannot trivially reduce to $0=0$. The growth of the Coriolis term forces an equivalent $\mathcal{O}(\alpha^{-1})$ growth in the cross-correlation of the pressure gradient, demonstrating how planetary rotation reorganizes the large-scale flow as already explained in Appendix \ref{app:heuristics}.

\vspace{8mm}

\textbf{Funding.}  A.H. is supported by the ERC Consolidator Grant ProbQuant (jointly with the Swiss State
Secretariat for Education, Research and Innovation).

G.S. is funded in part by the NSF grant DMS-2306378 and the Simons Foundation through the Simons Collaboration on Wave Turbulence.

\textbf{Conflict of interest.} The authors declare that they have no conflict of interest.

\textbf{Data availability.} Data sharing is not applicable.
We do not analyse or generate any datasets, because our work proceeds within a theoretical and mathematical approach.

\bibliographystyle{plain}
\bibliography{bib}  

\end{document}